\documentclass[lettersize,journal]{IEEEtran}

\usepackage{cite}

\ifCLASSINFOpdf
\usepackage[pdftex]{graphicx}
\graphicspath{{../pdf/}{../jpeg/}}
\DeclareGraphicsExtensions{.pdf,.jpeg,.png}
\else
\usepackage[dvips]{graphicx}
\graphicspath{{../eps/}}
\DeclareGraphicsExtensions{.eps}
\fi

\usepackage{amsmath}
\usepackage{amsfonts}
\usepackage{amssymb}
\usepackage{amsthm}
\usepackage{array}
\usepackage{float}
\usepackage[caption=false,font=footnotesize]{subfig}
\usepackage{dblfloatfix}
\usepackage{tikz}
\usetikzlibrary{shapes}
\usepackage{pgfplots}
\usepackage{xcolor}
\usepackage[outline]{contour}
\contourlength{3pt}
\usepackage{mleftright}
\mleftright
\usepackage{cuted}
\usepackage{flushend}
\usepackage{multirow}

\usepackage[normalem]{ulem}
\usepackage[nice]{nicefrac}

\newtheorem{thm}{Theorem}
\newtheorem{lemma}{Lemma}

\newtheorem{remark}{Remark}
\newtheorem{definition}{Definition}
\newtheorem{corollary}{Corollary}

\allowdisplaybreaks

\newcommand{\cc}[1]{\textcolor{black}{#1}}

\begin{document}
	
	\title{Capacity Bounds for the Two-User IM/DD Interference Channel}
	\author{Zhenyu~Zhang and Anas~Chaaban,~\IEEEmembership{Senior~Member,~IEEE}
		\thanks{
			Preliminary results related to this paper were presented in the 16th Canadian Workshop on Information Theory (CWIT) \cite{zhang2019capacity}.%
			
			Zhenyu (Charlus) Zhang and Anas Chaaban are with the School of Engineering, University of British Columbia, Kelowna, BC V1V 1V7, Canada, e-mail: zhenyu.zhang@alumni.ubc.ca; anas.chaaban@ubc.ca.
			
			This research was funded by the Natural Sciences and Engineering Research Council of Canada (NSERC) under grant RGPIN-2018-04254. 
	}}
	
	\maketitle
	
	\begin{abstract}
		This paper studies the capacity of the two-user intensity-modulation/direct-detection (IM/DD) interference channel (IC), which is relevant in the context of multi-user optical wireless communications. Despite some known single-letter capacity characterizations for general discrete-memoryless ICs, a computable capacity expression for the IM/DD IC is missing. In this paper, we provide tight and easily computable inner and outer bounds for a general two-user IM/DD IC under peak and average optical intensity constraints. The bounds enable characterizing the asymptotic sum-rate capacity in the strong and weak interference regimes, as well as the generalized degrees of freedom (GDoF) in the symmetric case. Using the obtained bounds, the GDoF of the IM/DD IC is shown to have a `W' shape similar to the Gaussian IC with power constraints. The obtained bounds are also evaluated numerically in different interference regimes to show their tightness, and used to study the performance of on-chip and indoor OWC systems.
	\end{abstract}
	
	\IEEEpeerreviewmaketitle
	
	\section{Introduction}
	
	
	Optical wireless communication (OWC) makes use of light-emitting devices (LEDs or laser diodes) as transmitters, and photo-detectors as receivers to realize wireless information transmission \cite{tsonev2014light,haas2015lifi,hranilovic2017trends}. It can be realized using collimated beams such as in free-space optical (FSO) communication \cite{khalighi2014survey}, or using diffused light such as in visible light communication (VLC). VLC integrates high-speed wireless communication into existing illumination systems, including home light fixtures, road lamps, and car head/tail-lights. Other applications of OWC that use diffused light include underwater OWC \cite{johnson2014recent} for fast and marine-friendly underwater data transmission, and OWC on-chip \cite{nafari2017chip} as a space-saving alternative to wire buses in multi-core chips. 
	The broadcast nature of OWC using diffused light leads to interference between transmitter-receiver pairs. Thus, it is important to study the impact of interference on the performance of OWC systems. To this end, an information-theoretic (IT) perspective can be helpful as it provides performance benchmarks and informs scheme design. 
	
	From an IT perspective, communications performance depends on the communication channel, characterized by a probability distribution of the channel output given the input \cite{book-EIT}. An OWC channel is modeled by the so-called IM/DD channel\footnote{Also known as the FSO intensity channel and optical direct detection channel with Gaussian post-detection noise~\cite{Lapidoth-Moser-Wigger-2009}.} \cite{Lapidoth-Moser-Wigger-2009}, which belongs to the family of Gaussian channels. Its input must be real and nonnegative, and must satisfy mean and peak intensity constraints, to model constraints on the average and peak optical intensity. Due to this, existing results on the well-studied Gaussian channel with an average electrical power constraint (second moment constraint) which models radio-frequency (RF) communications cannot be directly applied to IM/DD channels. This opens a wide research area on IT limits of IM/DD channels.
	
	The IM/DD point-to-point (P2P) channel which models a single-user communication link was studied in \cite{Lapidoth-Moser-Wigger-2009,farid2010capacity, chaaban2016free}, which derived capacity upper and lower bounds under average and/or peak intensity constraints. 
	The capacity of multiple input and/or multiple output IM/DD P2P channels was studied in \cite{moser2018capacity,li2020capacity,
	chaaban2018low,chaaban2018capacity}. 
	While single-user IM/DD channels have no interference, multi-user channels have interference that must be taken into account when studying their capacity. The capacity of some IM/DD multi-user channels has been studied in the literature, including the IM/DD broadcast channel (BC) modeling communication between a transmitter and multiple receivers, and the IM/DD multiple-access channel (MAC) modeling communication between multiple transmitters and a receiver \cite{chaaban2016capacity,chaaban2017capacity,zhou2019bounds}. However, the capacity of the IM/DD interference channel (IC), which models communications between multiple interfering transmitter-receiver pairs, has not been studied in the literature despite its importance in VLC systems~\cite{rahaim2017interference,fuschini2020multi}.

	An IM/DD IC can model an indoor VLC system where multiple light sources communicate with their respective receiver, leading to interference due to random user locations and overlap between coverage regions \cite{rahaim2017interference}. It can also model interference in the emerging area of on-chip OWC, where isolated chips communicate with each other wirelessly using light \cite{nafari2017chip,fuschini2020multi}. Similarly, in vehicular VLC, an IM/DD IC can model cars in different lanes exchanging information with cars in front or behind them. While in some scenarios transmitter-side coordination can be established to avoid interference using time-division multiple access (TDMA) e.g., this is generally sub-optimal \cite{Etkin-Tse-Wang-2008}. This highlights the importance of studying the capacity of the IM/DD IC which is the focus of this paper.
	
	Studying the capacity of the IC has seen research activity since Shannon's work in 1961 \cite{Shannon-1961} and until nowadays \cite{Han-kobyashi-1981,kramer2004outer,dytso2016interference}
	However, the capacity region of the IC is not known in general except for some special cases such as when interference is strong \cite{costa1987capacity}. Consequently, in order to uncover the secrets of the IC, researchers focused on the smallest IC, i.e., the two-user IC with two transmitter-receiver pairs. Treating-interference-as-noise (TIN) is the easiest way to deal with interference, but does not guarantee good performance. It is shown in \cite{Han-kobyashi-1981} that combining TIN with transmission of common messages decoded at both receivers can enlarge the achievable rate region relative to TIN, making the Han-Kobayashi (HK) scheme \cite{Han-kobyashi-1981} the best known scheme for the IC in terms of achievable rate. 
	A thorough study of the general two-user Gaussian IC under an average electrical power constraint (modeling RF networks) is presented in \cite{Etkin-Tse-Wang-2008}, which provides capacity inner and outer bounds with a gap less than one bit, and introduces the generalized degrees-of-freedom (GDoF) to approximate the IC capacity at high signal-to-noise ratio (SNR), where performance is mainly affected by interference. Since a closed-form expression for the IC capacity is hard to derive, the GDoF can be used as the prelog factor for approximating the IC capacity under various interference strengths \cite{davoodi2017generalized}. From this perspective, the GDoF can be seen as a generalization of the degrees-of-freedom (DoF)\footnote{Also known as the capacity prelog or the multiplexing gain \cite{cadambe2008interference}}, which captures the case when the interference and the desired signal are equally strong. The GDoF is shown to be an important metric in studying and designing communication schemes for more complex networks \cite{davoodi2017generalized}. To explicitly derive the GDoF, computable and fairly tight capacity bounds are required. 
	
	However, computable capacity bounds for the IM/DD IC do not exist to-date, even in the strong interference regime. Although the capacity region of the general IC with strong interference can be used to express the capacity region of the IM/DD IC with strong interference, this remains difficult to compute since it involves a union over distributions. This problem is compounded by the absence of a simple capacity expression for the IM/DD P2P channel.
	Note that achievable rates for the IM/DD IC were studied in \cite{ma2019capacity} which focused on TIN in the MIMO IM/DD IC. However, \cite{ma2019capacity} did not present capacity outer bounds for investigating the tightness of the achievable rate regions. Our preliminary work in \cite{zhang2019capacity} derived computable capacity bounds for a symmetric two-user IM/DD IC with an average intensity constraint. The more general asymmetric IM/DD IC with peak and average optical intensity constraints is yet to be studied.
	
	Due to the lack of computable expressions for the IM/DD IC and IM/DD P2P channel capacities, the GDoF of the IM/DD IC is also unknown. Recent research on the standard Gaussian IC \cite{dytso2014gaussian, dytso2016interference} shows that using pulse amplitude modulation (PAM) for the common signal combined with using Gaussian or PAM for the private signal is GDoF optimal. Since a DC-biased PAM signal can be used in an IM/DD channel, the inner bounds in \cite{dytso2014gaussian, dytso2016interference} apply for the IM/DD IC at a fixed average-to-peak optical intensity ratio of $\frac{1}{2}$. However, the same does not apply directly for an IM/DD IC under more general average and peak intensity constraints. Also, it is unknown if the obtained achievable GDoF in \cite{dytso2014gaussian, dytso2016interference} is optimal for the IM/DD IC since this requires deriving GDoF upper bounds for the IM/DD IC. 
	
	In this paper, results in \cite{zhang2019capacity} are extended to the general two-user IM/DD IC with general peak and average intensity constraints, aiming to find fairly tight and computable capacity bounds and to characterize the GDoF of the IM/DD IC. 
	Two inner bounds and two outer bounds are derived, where the inner bounds are based on TIN and HK schemes, and the outer bounds are derived by providing receiver side information to form Z-channels and a genie-aided channel \cite{Etkin-Tse-Wang-2008}. Numerical comparison are provided to show the asymptotic tightness of the obtained bounds at high SNR. The asymptotic tightness leads to a fairly tight approximation of the capacity region in the strong interference regime, a finite-gap approximation of the sum-rate capacity in the weak interference regime, and the GDoF of the symmetric IM/DD IC. Interestingly, the GDoF retains the `W' shape observed in the standard Gaussian IC in \cite{Etkin-Tse-Wang-2008}. It is worth to note that the obtained inner and outer bounds are based on a capacity upper bound for the IM/DD P2P channel from \cite{Lapidoth-Moser-Wigger-2009}, and hence can be tightened by using tighter IM/DD P2P channel capacity bounds. Finally, to demonstrate the applicability of the obtained results, we consider an on-chip OWC system and an indoor VLC system as examples, where we discuss system design problems such as optimizing the receivers' location or evaluating interference management capability of different schemes.
	
	The rest of the paper is organized as follows. After defining the two-user IM/DD IC model and listing the necessary preliminaries and definitions in Sec. \ref{Sec:Model}, Sec. \ref{Sec:Bounds} presents the obtained capacity inner and outer bounds for the IM/DD IC and shows their tightness through numerical evaluations at finite SNR. Then, the asymptotic tightness of the obtained capacity bounds is studied in Sec. \ref{Sec:Asymptotics}, and the GDoF of the symmetric IM/DD IC is obtained. The obtained results are then used for studying two example practical systems in Sec. \ref{Sec:Numerical}. Finally, Sec. \ref{Sec:Conclusion} concludes this paper. 
	
	
	Throughout the paper, we use bold-face letters to denote vectors, normal-face lower-case letters to denote scalars, and normal-face upper-case letters to denote random variables, unless otherwise specified. 
	All logarithms are base 2. 
	For a random variable $X$, we use $p(x)$, $\mathbb{E}[X]$, and $\mathsf{h}(X)$ to denote its probability density function, expectation, and differential entropy, 
	and a superscript-form $X^n$ to denote a time sequence $(X^{(t)})_{t=1}^{n}$ with distribution $p(x^n)$; $\mathsf{I}(X;Y)$ is the mutual information between $X$ and $Y$; 
	$\mathcal{N}(\mu,\sigma^2)$ represents the Gaussian distribution with mean $\mu$ and variance $\sigma^2$, and $\mathcal{Q}(x)=\int_{x}^{\infty}\frac{1}{\sqrt{2\pi}}e^{-\frac{x^2}{2}}dx$ denotes the standard Gaussian tail function. 
	We denote by $\lceil x \rceil$ the smallest integer no smaller than $x$, by ${\rm Conv}(\mathcal{S})$ the convex hull of a set $\mathcal{S}$, and by $\mathbb{R}_+$ the set of nonnegative real numbers.

	\section{Channel Model and Preliminaries}
	\label{Sec:Model}
	In this section, we introduce the two-user IM/DD IC model, in addition to some IM/DD P2P channel capacity bounds from \cite{Lapidoth-Moser-Wigger-2009} and other preliminaries that are required later in this paper. 
	
	\newcounter{mytempeqncnt}
	\begin{figure*}[h!b!]
		\normalsize
		\hrulefill
		\setcounter{mytempeqncnt}{\value{equation}}
		\setcounter{equation}{4}
		\begin{subequations}
			\label{eq:ci}
			\begin{align} 
				c_1 &=
				\min_{\mu,\xi>0} \biggl[ 1- \mathcal{Q} \Bigl( \frac{\xi + \alpha \mathsf{A}}{\sigma} \Bigr) - \mathcal{Q} \Bigl( \frac{\xi + (1-\alpha)\mathsf{A}}{\sigma} \Bigr)  \biggr] \log\bigg( \frac{\mathsf{A}}{\sigma} \frac{e^{ \frac{\mu \xi}{\mathsf{A}} } - e^{ -\mu(1+\frac{\xi}{\mathsf{A}}) }}{\sqrt{2\pi}\mu(1-2\mathcal{Q}(\frac{\xi}{\sigma}) )} \bigg) 
				\notag  \\ 
				& \qquad\qquad\qquad\qquad + \log(e) \biggl[ - \frac{1}{2} + \mathcal{Q} \Bigl( \frac{\xi}{\sigma} \Bigr) + \frac{\xi}{\sqrt{2\pi}\sigma} e^{ -\frac{\xi^2}{2\sigma^2} } + \frac{\sigma}{\mathsf{A}} \frac{\mu}{\sqrt{2\pi}} \Bigl( e^{ -\frac{\xi^2}{2\sigma^2} }  - e^{ -\frac{(\mathsf{A}+\xi)^2}{2\sigma^2} } \Bigr) 
				\notag  \\ 
				& \qquad\qquad\qquad\qquad + \mu\alpha \biggl( 1- 2 \mathcal{Q} \Bigl( \frac{\xi + \frac{\mathsf{A}}{2}}{\sigma} \Bigr)  \biggr) \biggr],  \\
				c_2 &=
				\min_{\xi>0}  \biggl[ 1- 2 \mathcal{Q} \Bigl( \frac{\xi + \frac{\mathsf{A}}{2}}{\sigma} \Bigr)  \biggr] 
				\log \biggl( \frac{\mathsf{A}+2\xi}{\sigma \sqrt{2\pi} (1- 2\mathcal{Q}(\frac{\xi}{\sigma}))} \biggr) 
				+ \log(e) \biggl[- \frac{1}{2} + \mathcal{Q} \Bigl( \frac{\xi}{\sigma} \Bigr) + \frac{\xi}{\sqrt{2\pi}\sigma} e^{ -\frac{\xi^2}{2\sigma^2} } \biggr].
			\end{align}
		\end{subequations}
		\setcounter{equation}{\value{mytempeqncnt}}
	\end{figure*}
	
	\subsection{The Two-User IM/DD Interference Channel (IC)} \label{sec:model}
	
	We consider a two-user IM/DD IC as shown in Fig. \ref{fig:chan}. In this channel, there are two transmitter-receiver pairs, where transmitter $i\in\{1,2\}$, communicates with receiver $i$ while interfering with receiver $j\in\{1,2\}$, $j\ne i$. At any time instant, the transmission can be modeled as
	\begin{align} \label{eq:model-channel}
	Y_i = h_{ii} X_i + h_{ji} X_j + Z_i,
	\end{align} 
	for $i,j \in \{1,2\}$ and $i\neq j$, where $Y_i$ and $X_i$ model the received and transmitted signals, respectively, $h_{ij} \in \mathbb{R}_+$ is the channel gain\footnote{$h_{ij}\geq 0$ since it is the ratio of the received to the transmitted optical intensity, both of which are nonnegative.} from transmitter $i$ to receiver $j$, and $Z_i$ is the noise at receiver $i$, which we assume to be $\mathcal{N}(0, \sigma_i^2)$ (thermal-noise-limited regime).
	
	Due to the physical properties of solid-state optical transmitters, $X_i$ should be real, nonnegative, and peak constrained. Although turning on an LED requires the injected current to be larger than a constant threshold, we assume the threshold to be zero without loss of generality when studying the channel capacity. Also, $X_i$ may be subject to an average optical intensity constraint for illumination and eye-safety reasons in applications such as indoor VLC. Thus, the following constraints are considered in studying the capacity of the IM/DD IC:
	\begin{enumerate}
		\item Real and nonnegative intensity constraint: $X_i\in\mathbb{R}_+$, $i=1,2$;
		\item Peak intensity constraint: $X_i\le \mathsf{A}_i$, $i=1,2$;
		\item Average intensity constraint: $\mathbb{E}[X_i]\le\mathsf{E}_i$, $i=1,2$.
	\end{enumerate}
	Without loss of generality, we assume that $\mathsf{A}_1=\mathsf{A}_2=\mathsf{A}$ and $\sigma_1=\sigma_2=\sigma$, since differences in $\mathsf{A}_i$ and $\sigma_i^2$ can be absorbed into the channel coefficients $h_{ij}$, $i,j\in\{1,2\}$. We define $\alpha_i=\frac{\mathsf{E}_i}{\mathsf{A}}$, which will be extensively used in the paper. We also define a symmetric IC as the one with $h_{11}=h_{22}=1$, $h_{12}=h_{21}=g$, and $\alpha_1=\alpha_2=\alpha$. It is also worth to mention that the setting that $\mathsf{A}_1=\mathsf{A}_2$ and $\alpha_1=\alpha_2$ with general channel coefficients $h_{ij}$ models an IM/DD IC with identical transmitters. 

	\begin{figure}[!tbp]
		\centering
		\resizebox {.9\columnwidth} {!} {
			\begin{tikzpicture}[node distance = 1.5cm, inner sep=0, outer sep=0] 
			\node (m1) []{$M_1$};
			\node (enc1) [draw, right of=m1, inner sep=4] {Enc. 1};
			\node (x1) [right of=enc1] {$X_1^n$};
			\node (osum1) [circle, draw,  right of =x1, node distance = 3cm] {\large $+$};
			\node (y1) [right of=osum1, node distance = 1cm] {$Y_1^n$};
			\node (z1) [above of=osum1, node distance = 1cm] {$Z_1^n$};
			\node (dec1) [draw, right of=y1, inner sep=4] {Dec. 1};
			\node (mh1) [right of=dec1] {$\hat{M}_1$};
			
			\node (m2) [below of=m1]{$M_2$};
			\node (enc2) [draw, below of=enc1, inner sep=4] {Enc. 2};
			\node (x2) [below of=x1] {$X_2^n$};
			\node (osum2) [circle, draw,  below of=osum1] {\large $+$};
			\node (z2) [below of=osum2, node distance = 1cm] {$Z_2^n$};
			\node (y2) [below of=y1] {$Y_2^n$};
			\node (dec2) [draw, below of=dec1, inner sep=4] {Dec. 2};
			\node (mh2) [below of=mh1] {$\hat{M}_2$};
			
			\draw [->] (m1) -- (enc1);
			\draw [->] (enc1) -- (x1);
			\draw [->] (x1) -- node[anchor = south] {$h_{11}$}(osum1);
			\draw [->] (z1) -- (osum1);
			\draw [->] (osum1) -- (y1);
			\draw [->] (y1) -- (dec1);
			\draw [->] (dec1) -- (mh1);
			\draw [->] (m2) -- (enc2);
			\draw [->] (enc2) -- (x2);
			\draw [->] (x2) -- node[anchor = north] {$h_{22}$}(osum2);
			\draw [->] (z2) -- (osum2);
			\draw [->] (osum2) -- (y2);
			\draw [->] (y2) -- (dec2);
			\draw [->] (dec2) -- (mh2);
			
			\draw [->] (x1) --  (osum2);
			\draw [->] (x2) -- (osum1) {};
			\node [below of=x1, node distance = 10, xshift = 30] {$h_{12}$};
			\node [above of=x2, node distance = 10, xshift = 30] {$h_{21}$};
			\end{tikzpicture}
		}
		\caption{A two-user interference channel (IC). Message $M_i$ of user $i$ is encoded into a codeword $X_i^n$ of length $n$, which is sent over the channel. Receiver $i$ receives $Y_i^n$ which it uses to decode $\hat{M}_i$. The transmission is considered successful if $\hat{M}_i=M_i$.}
		\label{fig:chan}
	\end{figure}
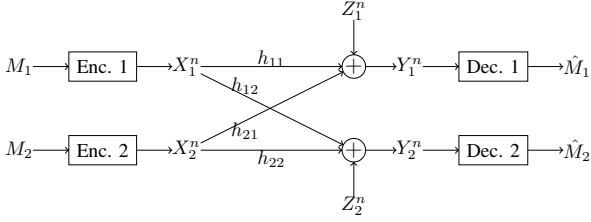
	
	Transmitter $i$ transmits a message $m_i \in \mathcal{M}_i$ to receiver $i$ through a length-$n$ codeword $x^n_i(m_i)=\bigl(x_i^{(1)}(m_i), \dots,x_i^{(n)}(m_i)\bigr)$ chosen from a codebook $\mathcal{C}_{i,n}$, where $\mathcal{M}_i = \bigl\{1, \ldots, \lceil 2^{n \mathsf{R}_i} \rceil  \bigr\}$ is the message set with code rate $\mathsf{R}_i$. Each element of the codeword is a realization of the random variable $X_i$ that satisfies the input constraints.  
	Receiver $i$ uses a decoding function $\varphi_{i,n}: \mathbb{R}^n \mapsto \mathcal{M}_i$ to obtain $\hat{m}_i = \varphi_{i,n} (y^n_i)$, and an error occurs if $\hat{m}_i \ne m_i$. The average error probability for receiver $i$ is defined as $\epsilon_{i,n} = \frac{1}{\lceil 2^{n \mathsf{R}_i} \rceil}\sum_{M_i\in\mathcal{M}_i} \mathbb{P} (\hat{M}_i \ne M_i) $. A rate pair $(\mathsf{R}_1, \mathsf{R}_2)$ is achievable if and only if there exists a sequence of codebook pairs $\left\{(\mathcal{C}_{1,n}, \mathcal{C}_{2,n})\right\}_n$ and decoding functions $\left\{(\varphi_{1,n}, \varphi_{2,n})\right\}_n$, such that $\epsilon_{1,n},\epsilon_{2,n} \rightarrow 0$, as $n \rightarrow \infty$. We denote by $\mathcal{G}$ the capacity region which is the closure of the set of all achievable rate pairs, and define the sum-rate capacity as $\mathsf{C}_\text{sum} \triangleq \max_{(\mathsf{R}_1,\mathsf{R}_2) \in \mathcal{G}} \{\mathsf{R}_1+\mathsf{R}_2\}$. 
	
	Next, we introduce some IM/DD P2P channel capacity bounds and other preliminaries that will be used throughout the paper.
	
	\subsection{Preliminaries} \label{sec:p2p-bounds}

	We start with IM/DD P2P channel capacity bounds. Consider an IM/DD P2P channel with input $X\in[0,\mathsf{A}]$ satisfying $\mathbb{E}[X]\le\alpha\mathsf{A}$, and output $Y=X+Z$ where $Z\sim \mathcal{N}(0, \sigma^2)$. Denote the IM/DD P2P channel capacity by $\mathsf{C}(\mathsf{A},\alpha) = \sup_{p(x):X\in[0,\mathsf{A}],\mathbb{E}[X]\le\alpha\mathsf{A}} \mathsf{I}(X; Y)$. The following lemma states lower and upper bounds for $\mathsf{C}(\mathsf{A},\alpha)$. 
	
	\begin{lemma}[\cite{Lapidoth-Moser-Wigger-2009}]
		\label{lemma:Cp2p}  
		$\mathsf{C}(\mathsf{A},\alpha)$ satisfies $\underline{\mathsf{C}} (\mathsf{A}, \alpha) \le \mathsf{C}(\mathsf{A},\alpha) \le \overline{\mathsf{C}} (\mathsf{A}, \alpha)$, where $\underline{\mathsf{C}} (\mathsf{A}, \alpha)$ and $\overline{\mathsf{C}} (\mathsf{A}, \alpha)$ are defined in \eqref{eq:c0_lb} and \eqref{eq:c0_ub}, respectively, as follows:
		\begin{equation}\label{eq:c0_lb}
		\underline{\mathsf{C}} (\mathsf{A}, \alpha) = \frac{1}{2} \log \left( 1 + \frac{\rho(\alpha)\mathsf{A}^2}{\sigma^2} \right) 
		\end{equation}
		wherein 
		\begin{equation} \label{eq:rho}
		\rho(\alpha) =\begin{cases}
		\frac{e^{2\alpha\mu^*-1}}{2\pi}\big(\frac{1-e^{-\mu^*}}{\mu^*}\big)^2, & \alpha\in(0,\frac{1}{2}),\\
		\frac{1}{2\pi e}, &  \alpha\in[\frac{1}{2},1],		
		\end{cases}
		\end{equation}
		and $\mu^*$ is the unique positive solution of $\alpha = \frac{1}{\mu^*} - \frac{e^{-\mu^*}}{1-e^{-\mu^*}}$; and
		\begin{equation} \label{eq:c0_ub}
		\overline{\mathsf{C}} (\mathsf{A}, \alpha) =\begin{cases}
		c_1, & \alpha\in(0,\frac{1}{2}),\\
		c_2, & \alpha\in[\frac{1}{2},1],
		\end{cases}
		\end{equation}
		wherein $c_1$ and $c_2$ are defined in \eqref{eq:ci}.
	\end{lemma}

	The detailed proof of Lemma \ref{lemma:Cp2p} can be found in \cite{Lapidoth-Moser-Wigger-2009}, where we note that the lower bound $\underline{\mathsf{C}} (\mathsf{A}, \alpha)$ is obtained by choosing $X$ to follow $\mathsf{p}(x;\mathsf{A},\alpha)$, as defined in \eqref{eq:p} below, which maximizes $\mathsf{h}(X)$ under the constraints $X\in[0,\mathsf{A}]$ and $\mathbb{E}[X]\le\alpha\mathsf{A}$ \cite{Lapidoth-Moser-Wigger-2009}. 
	\begin{equation}\label{eq:p}
		\mathsf{p}(x;\mathsf{A},\alpha)=\begin{cases}
		\frac{1}{ \mathcal{\mathsf{A}} } \frac{\mu^*}{1- e^{-\mu^*}} e^{-\frac{\mu^* x}{\mathsf{A}}}, & \alpha\in(0,\frac{1}{2}), \\
		\frac{1}{\mathsf{A}}, & \alpha\in[\frac{1}{2},1].
		\end{cases}
	\end{equation}
	For $X\sim\mathsf{p}(x;\mathsf{A},\alpha)$, the corresponding mean and differential entropy of $X$ are given as follows:
	\begin{subequations}
		\begin{align}
			\mathbb{E}[X] &= \begin{cases}
			\alpha \mathsf{A}, & \alpha\in(0,\frac{1}{2}), \\
			\frac{\mathsf{A}}{2}, & \alpha\in[\frac{1}{2},1], 
			\end{cases} \\
			\mathsf{h}(X) &= \begin{cases}
			\substack{\log(e) \mu^* \bigl( \frac{1}{\mu^*} - \frac{e^{-\mu^*}}{1-e^{-\mu^*}} \bigr) \\ \hspace{2.5em}-\log \bigl( \frac{1}{\mathsf{A}} \frac{\mu^*}{1-e^{-\mu^*}} \bigr), } & \alpha\in(0,\frac{1}{2}),\\
			\log(\mathsf{A}), & \alpha\in[\frac{1}{2},1],
			\end{cases} \\
			& = \frac{1}{2}\log\bigl( 2\pi e\rho(\alpha)\mathsf{A}^2 \bigr).\label{h_of_X_A_alpha}
		\end{align}
	\end{subequations}
	Additionally, the following properties are required in the asymptotic analysis in Sec. \ref{Sec:Asymptotics}.
	\begin{remark}\label{rmk:property}
		The lower and upper bounds of $\mathsf{C}(\mathsf{A},\alpha)$ in Lemma \ref{lemma:Cp2p} satisfy the following properties:
		\begin{itemize}
			\item[(i)] $\varepsilon(\mathsf{A}, \alpha)\triangleq\overline{\mathsf{C}}(x,\alpha)-\underline{\mathsf{C}}(x,\alpha)< 0.7$ bits for all $\mathsf{A}\in\mathbb{R}_+$, and $\lim_{\frac{\mathsf{A}}{\sigma} \rightarrow \infty}\varepsilon(\mathsf{A}, \alpha) = 0$;
			\item[(ii)] $0<\rho(\alpha)\le\frac{1}{2\pi e}$ and $\rho(\alpha)$ is monotonically nondecreasing;
			\item[(iii)] $\underline{\mathsf{C}}(\mathsf{A},\alpha)$ is monotonically increasing in $\mathsf{A}$, $\forall \alpha\in(0,1]$;
		\end{itemize}
	\end{remark}	
The proof of (i) can be found in \cite{Lapidoth-Moser-Wigger-2009} which also states that $\varepsilon(\mathsf{A}, \alpha)<1\; \textnormal{nat}\approx0.7\;\textnormal{bits}$, obtained through numerical evaluation; (ii) follows since $2^{2\mathsf{h}(X)}=2\pi e\rho(\alpha)\mathsf{A}^2$ and since $h(X)\le\log(A)$ \cite[Example 12.2.4]{book-EIT}, and the monotonicity of $\rho(\alpha)$ is obtained numerically; and (iii) follows since the logarithm in \eqref{eq:c0_lb} is monotonically increasing.

	Next, we provide some definitions that will be needed in the sequel. Let 	$\boldsymbol{\kappa}=(\kappa_1,\kappa_2)$, $\boldsymbol{\eta}=(\eta_1,\eta_2)$, and $\boldsymbol{\theta}=(\theta_1,\theta_2)$, wherein $\kappa,\eta_i\in\mathbb{R}_+$,  $\theta_i\in(0,1]$, $i=1,2$.

\begin{definition}[Feasible input distributions]\label{Def:TIN_Feasible_Dist}
	For $\mathbf{x}=(x_1,x_2)$, the set $\mathcal{P}(\boldsymbol{\kappa},\boldsymbol{\theta})$ of feasible input distributions for the IM/DD IC is defined as the set of distributions $p(\mathbf{x})=p(x_1)p(x_2)$ satisfying $p(\mathbf{x})=0$ if $\mathbf{x}\notin[0,\kappa_1]\times [0,\kappa_2]$ and $\mathbb{E}[X_i] \le \theta_i\kappa_i$, $i=1,2$.
\end{definition}

\begin{definition}[Feasible distributions for the Han-Kobayashi (HK) Scheme]\label{Def:HK_Feasible_Dist}
	For $\mathbf{u}=(u_1,u_2)$, and $\mathbf{w}=(w_1,w_2)$, the set $\mathcal{P}_{\rm HK}(\boldsymbol{\kappa}, \boldsymbol{\eta}, \boldsymbol{\theta})$ of feasible input distributions for a HK scheme in the IM/DD IC is defined as the set of distributions $p(\mathbf{u},\mathbf{w})= \prod_{i=1}^{2}p(u_i)p(w_i)$ satisfying $p(\mathbf{u},\mathbf{w})=0$ if $(\mathbf{u},\mathbf{w})\notin[0,\kappa_1]\times[0,\kappa_2]\times[0,\eta_1]\times[0,\eta_2]$, and $\mathbb{E}[U_i+W_i] \le \theta_i(\kappa_i+\eta_i)$. 
\end{definition}

Then, Lemma \ref{Lem:FunctionF} and Corollary \ref{Cor:FunctionF} below state a mutual-information lower bound that will be used to bound the achievable rate of decoding a pair of messages from a received signal with noise and interference. Lemma \ref{Lem:FunctionF} and Corollary \ref{Cor:FunctionF} will help to derive the achievable rates of TIN and the HK scheme in Sec. \ref{Sec:Bounds}. 

\begin{lemma}\label{Lem:FunctionF}
Consider independent random variables $V_i\in[0,a_i]$, $i=1,2,3$ satisfying $\mathbb{E}[V_i]\le\theta_ia_i$ where $a_i\in\mathbb{R}_+$ is the peak intensity parameter and $\theta_i\in(0,1]$ is the average-to-peak ratio parameter, and let $\bar{V}=V_1+V_2+V_3+Z$ where $Z\sim\mathcal{N}(0,\sigma^2)$ is independent of $V_i$. Then, if $V_i$, $i=1,2,$ is distributed according to $\mathsf{p}(v_i;a_i,\theta_i)$ (cf. \eqref{eq:p}), we have
\begin{align}\label{eq:F}
	\mathsf{I}(V_1, V_2; \bar{V}) &\ge \mathsf{F}(a_1,a_2,a_3,\theta_1,\theta_2,\theta_3) \notag \\
	&\triangleq \frac{1}{2} \log \left(  1 + \frac{ \rho(\theta_1)a_1^2+\rho(\theta_2)a_2^2 } {\sigma^2 2^{ 2 \overline{\mathsf{C}}(a_3, \theta_3) } } \right).
\end{align}
\end{lemma}
\begin{proof}
We have
\begin{subequations}
	\begin{align}
		\mathsf{I}(V_1,V_2; \bar{V}) &= \mathsf{h}(V_1 + V_2 + V_3 + Z) - \mathsf{h}(V_3 + Z)\\
		&\ge \frac{1}{2} \log\left( 1+ \frac{2^{ 2 \mathsf{h}(V_1) }+2^{ 2 \mathsf{h}(V_2) }}{2^{ 2 \mathsf{h}( V_3+Z )} }\right)\\
		&= \frac{1}{2} \log\left( 1+ \frac{2^{ 2 \mathsf{h}(V_1) }+2^{ 2 \mathsf{h}(V_2) }}{2^{ 2 \mathsf{I}(V_3;V_3+Z)+2 \mathsf{h}(Z)} }\right),
	\end{align}
\end{subequations}
where the inequality follows using the entropy-power inequality (EPI) to lower-bound $\mathsf{h}(V_1 + V_2 + V_3 + Z)$ by $\frac{1}{2} \log\left(  2^{ 2 \mathsf{h}(V_1)}+2^{ 2 \mathsf{h}(V_2)} + 2^{ 2 \mathsf{h}( V_3+Z )}\right)$. 
Then, recall that $\mathsf{I}(V_3;V_3+Z)\leq \overline{\mathsf{C}}( a_3,\theta_3)$ by Lemma \ref{lemma:Cp2p} and $\mathsf{h}(Z)=\frac{1}{2}\log(2\pi e \sigma^2)$. Also, $2^{2\mathsf{h}(V_i)}=2\pi e \rho(\theta_i)a_i^2$ (cf. \eqref{h_of_X_A_alpha}) since $V_i\sim\mathsf{p}(v_i;a_i,\theta_i)$, $i=1,2$. Thus,
\begin{align}
\mathsf{I}(V_1,V_2; \bar{V})
&\geq \frac{1}{2} \log\left( 1+ \frac{\rho(\theta_1)a_1^2 + \rho(\theta_2)a_2^2 }{\sigma^2 2^{ 2 \overline{\mathsf{C}}(a_3,\theta_3)} }\right),
\end{align}
which proves the lemma.
\end{proof}
The following corollary is a direct result of Lemma \ref{Lem:FunctionF}.
\begin{corollary}\label{Cor:FunctionF}
$\mathsf{I}(V_2; V_2+V_3+Z) \ge \mathsf{F}(0,a_2,a_3,1,\theta_2,\theta_3)$, where $V_2$, $V_3$, $Z$, $a_i$, and $\theta_i$ are as defined in Lemma~\ref{Lem:FunctionF}.
\end{corollary}

Lastly, regarding the function $\mathsf{F}(a_1,a_2,a_3,\theta_1,\theta_2,\theta_3)$ in \eqref{eq:F}, it is worth to mention two special cases that will be frequently used in Sec. \ref{Sec:Asymptotics}:
\begin{enumerate}
	\item when $a_1,a_3=0$, $F(0,a_2,0,\theta_1,\theta_2,\theta_3)=\underline{\mathsf{C}}(a_2,\theta_2)$, and 
	\item when $\frac{a_3}{\sigma}\to\infty$, 
	\begin{multline}\label{F_Asymptotic}
		\lim_{\frac{a_3}{\sigma}\to\infty}\bigg[ \mathsf{F}(a_1,a_2,a_3,\theta_1,\theta_2,\theta_3) \\  - \frac{1}{2} \log \left(  1 + \frac{\rho(\theta_1)a_1^2+\rho(\theta_2)a_2^2 } {\rho(\theta_3)a_3^2 } \right) \bigg]=0,
	\end{multline}
	which follows since $\lim_{\frac{a_3}{\sigma} \rightarrow \infty}\varepsilon (a_3, \theta_3) = 0$ (Property (i)).
\end{enumerate} 

In the next section, we derive computable capacity inner and outer bounds for the IM/DD IC based on the preliminaries in this section.

	\section{IM/DD IC Capacity Bounds}
	\label{Sec:Bounds}
	In this section, computable capacity inner and outer bounds for the two-user IM/DD IC are presented. The inner bounds are obtained by TIN and HK schemes, and the outer bounds are obtained by providing side information to the IC to form Z-channels and a genie-aided IC. We start with deriving the two inner bounds.

	\subsection{Computable Inner Bounds}
	
	\subsubsection{Treating Interference as Noise (TIN)}
	A simple scheme to deal with interference in the two-user IC is to treat it as noise. Although it is a passive scheme, it can be efficient when interference is weak, i.e., when the cross-talk channel gain is small. In TIN, each receiver only decodes its desired signal, i.e., decoder $i$ only decodes $M_i$ while treating $X_j^n$ ($j \ne i$) as noise. Through TIN, rate pairs in the following set are achievable:
	\begin{align} \label{eq:TIN-R}
	\mathcal{R}_\text{TIN}\big(p(\mathbf{x})\big) \triangleq \left\lbrace (\mathsf{R}_1, \mathsf{R}_2)\in\mathbb{R}_+^2 \Big| 	\mathsf{R}_i \le \mathsf{I}(X_i ; Y_i), i\in\{1,2\} \right\rbrace,
	\end{align}
	where $\mathbf{x}=(x_1,x_2)$ and $p(\mathbf{x}) \in \mathcal{P}(\boldsymbol{\mathsf{A}},\boldsymbol{\alpha})$ defined in Def.  \ref{Def:TIN_Feasible_Dist}, wherein $\boldsymbol{\mathsf{A}}=(\mathsf{A},\mathsf{A})$. 
	To obtain the achievable rate region of TIN, one has to evaluate the mutual information in \eqref{eq:TIN-R} and take the union with respect to all feasible input probability distributions, which is intractable. Next, we simplify this task by providing an easily computable inner bound based on \eqref{eq:TIN-R}. 
	
	\begin{thm} \label{thm:TIN-innerB}
		The capacity region of the IC is inner bounded by $\mathcal{G}'_\text{TIN} \subseteq \mathcal{G}$, where 
		\begin{equation}\label{Gprime_TIN}
		\mathcal{G}'_\text{TIN} \triangleq {\rm Conv} \Biggl( \bigcup_{\substack{ \kappa_i\in[0,\mathsf{A}],\\\theta_i\kappa_i\in[0,\alpha_i\mathsf{A}],\; i=1,2}} \mathcal{R}'_\text{TIN}(\boldsymbol{\kappa},\boldsymbol{\theta}) \Biggr) ,
		\end{equation}
		wherein $\boldsymbol{\kappa}=(\kappa_1,\kappa_2)$ is the allocated peak intensities at the transmitters, $\boldsymbol{\theta}=(\theta_1,\theta_2)$, and 
		\begin{multline}\label{eq:TIN-R'}
		\mathcal{R}'_\text{TIN}(\boldsymbol{\kappa},\boldsymbol{\theta}) \triangleq  \Big\{(\mathsf{R}_1, \mathsf{R}_2)\in\mathbb{R}_+^2 \Big|\\	 \mathsf{R}_i \le \mathsf{F}(0, h_{ii}\kappa_i, h_{ji}\kappa_j, 1,\theta_i,\theta_j ), i,j\in\{1,2\}, i\neq j \Bigr\}.
		\end{multline}
	\end{thm}
	\begin{proof}
		
		For some $\kappa_i,\theta_i$, $i=1,2$, such that $\kappa_i\in[0,\mathsf{A}]$ and $\theta_i\kappa_i\in[0,\alpha_i\mathsf{A}]$, $i\in\{1,2\}$, let $\mathsf{R}_1=\mathsf{I}(X_1; Y_1)$ and $\mathsf{R}_2=\mathsf{I}(X_2; Y_2)$ with $(X_1, X_2)\sim p(\mathbf{x})=\prod_{i=1}^{2}\mathsf{p}(x_i;\kappa_{i},\theta_i)$. As $p(x)\in \mathcal{P}(\boldsymbol{\mathsf{A}},\boldsymbol{\alpha})$, we have $(\mathsf{R}_1,\mathsf{R}_2)\in\mathcal{R}_\text{TIN}\big(p(\mathbf{x})\big)$. 
		On the other hand, $\tilde{X}_{11}\triangleq h_{11}X_1\sim \mathsf{p}(\tilde{x}_{11};h_{11}\kappa_1,\theta_1)$ and $\tilde{X}_{21}\triangleq h_{21}X_2\sim \mathsf{p}(\tilde{x}_{21};h_{21}\kappa_2,\theta_2)$. Using Corollary \ref{Cor:FunctionF}, we obtain
		\begin{align} 
			\mathsf{R}_1 & = \mathsf{I}(\tilde{X}_{11};Y_1) \geq \mathsf{F}(0,h_{11}\kappa_1, h_{21}\kappa_2, 1,\theta_1,\theta_2),
		\end{align}
		and, similarly, $\mathsf{R}_2 \ge \mathsf{F}(0,h_{22}\kappa_2, h_{12}\kappa_1, 1,\theta_2,\theta_1)$. Since $(\mathsf{R}_1,\mathsf{R}_2)$ are achievable given any $\kappa_i,\theta_i$, $i=1,2$, such that $\kappa_i\in[0,\mathsf{A}]$ and $\theta_i\kappa_i\in[0,\alpha_i\mathsf{A}]$, then rate pairs in $\mathcal{R}'_\text{TIN}(\boldsymbol{\kappa},\boldsymbol{\theta})$ are all achievable, which concludes the proof.
	\end{proof}
	
	\subsubsection{Han-Kobayashi (HK) Scheme}
	A general Han-Kobayashi scheme for the two-user IC makes use of private and common messages together with superposition encoding and joint decoding to achieve a larger rate region than TIN \cite{Han-kobyashi-1981}. Here, we study the rate region achieved by a simplified HK scheme \cite{Etkin-Tse-Wang-2008} for the two-user IM/DD IC, which transmits the sum of the private and common messages and decodes successively. Specifically, in the simplified HK scheme (Fig. \ref{fig:HK}), each message $M_i \in \mathcal{M}_i$ is split into two parts, a private message $M_i^{\rm p} \in \mathcal{M}_i^{\rm p}$ with rate $\mathsf{R}_i^{\rm p}$ and a common message $M_i^{\rm c} \in \mathcal{M}_i^{\rm c}$ with rate $\mathsf{R}_i^{\rm c}$, $i=1,2$, so that $\mathsf{R}_i = \mathsf{R}_i^{\rm p} + \mathsf{R}_i^{\rm c}$. Then, the messages $M_i^{\rm p}$ and $M_i^{\rm c}$ are encoded independently into $U_i^n$ and $W_i^n$, respectively, where $U_i^n$ is decoded at receiver $i$ but $W_i^n$ is decoded at both receivers. 
	We choose $(\mathbf{U},\mathbf{W})\sim p(\mathbf{u},\mathbf{w})\in\mathcal{P}_{\rm HK}(\boldsymbol{\kappa},\boldsymbol{\eta},\boldsymbol{\alpha})$ defined in Def. \ref{Def:HK_Feasible_Dist} with $\kappa_i+\eta_i\le\mathsf{A}$.	
	Then, the channel input is constructed as $X_i=U_i + W_i$ resulting in an output $Y_i= h_{ii}(U_i + W_i) + h_{ji}(U_j + W_j) + Z_i$, $i \ne j$. Receiver $i$ obtains $M_1^{\rm c}$ and $M_2^{\rm c}$ first using joint decoding, then reconstructs $W_1$ and $W_2$ and subtracts them from $Y_i$, and finally uses the obtained signal for decoding $U_i$ while treating $U_j$ ($j \ne i$) as noise. 
	\begin{figure}[!tbp]
		\centering
		\resizebox {\columnwidth} {!} { 
			\begin{tikzpicture}[inner sep=0, outer sep=0]
			\node(x1) at (1,0) {$X_1$};
			\node(x2) at (1,-2) {$X_2$};
			\node(y1) at (3.5,0) {$Y_1$};
			\node(y2) at (3.5,-2) {$Y_2$};
			\node(z1) at (2.7,1) {$Z_1$};
			\node(z2) at (2.7,-3) {$Z_2$};
			\node (p1) at (2.7,0) [circle, draw] {\large $+$};
			\node (p2) at (2.7,-2) [circle, draw] {\large $+$};
			\node [below of=x1, node distance = 15, xshift = 10] {$h_{12}$};
			\node [above of=x2, node distance = 15, xshift = 10] {$h_{21}$};
			
			\draw[->] (x1) to node[above]{$h_{11}$} (p1);
			\draw[->] (x1) to (p2);
			\draw[->] (x2) to (p1);
			\draw[->] (x2) to node[above]{$h_{22}$} (p2);
			\draw[->] (z1) to (p1);
			\draw[->] (z2) to (p2);
			\draw[->] (p1) to (y1);
			\draw[->] (p2) to (y2);
			
			\node(e1p) at (-2.5,.5) [rectangle, draw, inner sep=4] {Enc. 1-p};
			\node(e2p) at (-2.5,-1.5) [rectangle, draw, inner sep=4] {Enc. 2-p};
			\node(e1c) at (-2.5,-.5) [rectangle, draw, inner sep=4] {Enc. 1-c};
			\node(e2c) at (-2.5,-2.5) [rectangle, draw, inner sep=4] {Enc. 2-c};
			\node(m1p) at (-4,.5) {$M_1^{\rm p}$};
			\node(m1c) at (-4,-.5) {$M_1^{\rm c}$};
			\node(m2p) at (-4,-1.5) {$M_2^{\rm p}$};
			\node(m2c) at (-4,-2.5) {$M_2^{\rm c}$};
			\draw[->] (m1p) to (e1p);
			\draw[->] (m1c) to (e1c);
			\draw[->] (m2p) to (e2p);
			\draw[->] (m2c) to (e2c);
			\node(u1) at (-1,.5) {$U_1^n$};
			\node(u2) at (-1,-1.5) {$U_2^n$};
			\node(w1) at (-1,-.5) {$W_1^n$};
			\node(w2) at (-1,-2.5) {$W_2^n$};
			\draw[->] (e1p) to (u1);
			\draw[->] (e2p) to (u2);
			\draw[->] (e1c) to (w1);
			\draw[->] (e2c) to (w2);
			
			\node (sum1) at (0,0) [circle, draw] {\large $+$};
			\node (sum2) at (0,-2) [circle, draw] {\large $+$};
			\draw[->] (u1) to (sum1);
			\draw[->] (u2) to (sum2);
			\draw[->] (w1) to (sum1);
			\draw[->] (w2) to (sum2);
			\draw[->] (sum1) to (x1);
			\draw[->] (sum2) to (x2);
			
			\node(d1) at (5,0) [rectangle, draw, inner sep=4] {Dec. 1};
			\node(d2) at (5,-2) [rectangle, draw, inner sep=4] {Dec. 2};
			\node(m1h) at (7.5,0) {$\hat{M}_1^{\rm p},(\hat{M}_1^{\rm c},\hat{M}_2^{\rm c})$};
			\node(m2h) at (7.5,-2) {$\hat{M}_2^{\rm p},(\hat{M}_1^{\rm c},\hat{M}_2^{\rm c})$};
			\draw[->] (y1) to (d1);
			\draw[->] (y2) to (d2);
			\draw[->] (d1) to (m1h);
			\draw[->] (d2) to (m2h);
			\end{tikzpicture}	
		}
		\caption{In the simplified Han-Kobayashi (HK) scheme, message $M_i$ is split into a private part $M_i^{\rm p}$ and a common part $M_i^{\rm c}$. These are encoded individually and their sum is transmitted. Each receiver decodes both common and one private message.}
		\label{fig:HK}
	\end{figure}
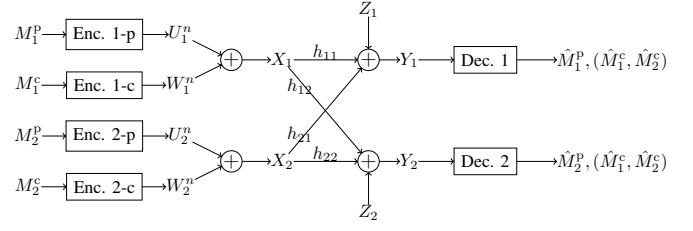
	
	Using the simplified HK scheme with distribution $p(\mathbf{u},\mathbf{w}) \in \mathcal{P}_{\rm HK}(\boldsymbol{\kappa},\boldsymbol{\eta},\boldsymbol{\alpha})$, the following rate constraints must be satisfied for reliable transmission, for $i,j\in\{1,2\}$ and $i\neq j$,
	\begin{subequations} 
		\label{eq:HK-R}
		\begin{IEEEeqnarray}{rCll}
			\mathsf{R}_i^{\rm u} &\le& \mathsf{I}(U_i; Y_i| W_1, W_2), \\
			\mathsf{R}_i^{\rm w} &\le& \min\bigl\{\mathsf{I}(W_i; Y_i | W_j), \mathsf{I}(W_i; Y_j | W_j)\bigr\},	\\
			\mathsf{R}_1^{\rm w}+\mathsf{R}_2^{\rm w}&\le& \min\bigl\{\mathsf{I}(W_1, W_2; Y_1), \mathsf{I}(W_1, W_2; Y_2)\bigr\}. 
		\end{IEEEeqnarray}
	\end{subequations}
	Thus, the following rate region is achievable: 
	\begin{multline}\label{eq:HK-R-all}
		\mathcal{R}_\text{HK}\big(p(\mathbf{u},\mathbf{w})\big) \triangleq \Bigl\{ (\mathsf{R}_1, \mathsf{R}_2)\in\mathbb{R}_+^2 \Big| \mathsf{R}_i = \mathsf{R}_i^{\rm u} + \mathsf{R}_i^{\rm w}, i\in\{1,2\}, \\ \eqref{eq:HK-R} \text{ is satisfied} \Bigr\}.
	\end{multline}
	The following theorem states a computable inner bound based on \eqref{eq:HK-R-all}. 

	\begin{thm} \label{thm:HK-innerB}
			The capacity region of the IC is inner bounded by $\mathcal{G}'_{\rm HK} \subseteq \mathcal{G}$, where 
			\begin{equation} \label{GprimeHK}
				\mathcal{G}'_{\rm HK} \triangleq {\rm Conv} \Biggl( \bigcup \mathcal{R}'_\text{HK} (\boldsymbol{\kappa}, \boldsymbol{\eta}, \boldsymbol{\theta}, \boldsymbol{\phi}) \Biggr) ,
				\end{equation}
			where the union is with respect to peak intensities allocated to the private and common messages, $\boldsymbol{\kappa}=(\kappa_1,\kappa_2)$ and  $\boldsymbol{\eta}=(\eta_1,\eta_2)$, respectively, and average-to-peak ratios of the private and common messages, $\boldsymbol{\theta}=(\theta_1, \theta_2)$ and $\boldsymbol{\phi}=(\phi_1,\phi_2)$, respectively, such that, for $i=1,2$, 
			\begin{subequations}
				\begin{align*} 
					\kappa_i,\eta_i&\in[0,\mathsf{A}],\\
					\kappa_i+\eta_i&\le\mathsf{A}, \\
					\theta_i\kappa_i+\phi_i\eta_i&\le\alpha_i\mathsf{A}
				\end{align*}
			\end{subequations}
			and $\mathcal{R}'_\text{HK} (\boldsymbol{\kappa}, \boldsymbol{\eta}, \boldsymbol{\theta}, \boldsymbol{\phi})$ is the set of rate pairs $(\mathsf{R}_1, \mathsf{R}_2)\in\mathbb{R}_+^2$ satisfying that, for $i,j\in\{1,2\}$, $i\ne j$,
			\begin{subequations}
				\label{eq:HK-R'}
				\begin{align}
					\mathsf{R}_i \le &\;\mathsf{F}(0,h_{ii}\kappa_i, h_{ji}\kappa_j, 1,\theta_i,\theta_j)\notag\\
					& {+}\min_{k\in\{1,2\}} \mathsf{F}(0, h_{ik}\eta_i, h_{ik}\kappa_i {+} h_{jk}\kappa_j, 1,\phi_i,\theta_{k}' ), \label{eq:HK-R'-1}\\
					\mathsf{R}_1{+}\mathsf{R}_2 &\le \sum_{\substack{i,j\in\{1,2\},\\i\ne j}} \mathsf{F}(0,h_{ii}\kappa_i, h_{ji}\kappa_j, 1,\theta_i,\theta_j)   \notag \\
					& \;\;\; {+}\min_{\substack{i,j\in\{1,2\},\\i\ne j}} \mathsf{F}( h_{ii}\eta_i, h_{ji}\eta_j, h_{ii}\kappa_i {+} h_{ji}\kappa_j, \phi_i,\phi_j,\theta_i' ), \label{eq:HK-R'-2}
				\end{align}
			\end{subequations}
			where $\theta_i'=\frac{h_{ii}\theta_i\kappa_i+h_{ji}\theta_j\kappa_j}{h_{ii}\kappa_i+h_{ji}\kappa_j}$ if $h_{ii}\kappa_i+h_{ji}\kappa_j>0$ and $\theta_i'=0$ otherwise, $i\neq j$.
		\end{thm}
	\begin{proof}
		Choose $(\kappa_i, \eta_i, \theta_i, \phi_i)$, $i=1,2$, such that the constraints in Theorem \ref{thm:HK-innerB} are satisfied and let $U_i\sim \mathsf{p}(u_i;\kappa_i,\theta_i)$ and $W_i\sim \mathsf{p}(w_i;\eta_i,\phi_i)$, which leads to 
		$p(\mathbf{u},\mathbf{w})\in\mathcal{P}_{\rm HK}(\boldsymbol{\kappa}, \boldsymbol{\eta},\boldsymbol{\alpha})$. Then, use \eqref{eq:HK-R} with $p(\mathbf{u},\mathbf{w})$ to obtain $\mathsf{R}_i^{\rm u}$ and $\mathsf{R}_i^{\rm w}$ that satisfy \eqref{eq:HK-R} with equality. Let $\mathsf{R}_i = \mathsf{R}_i^{\rm u} + \mathsf{R}_i^{\rm w}$. Then, we have $(\mathsf{R}_1,\mathsf{R}_2) \in \mathcal{R}_\text{HK}\big(p(\mathbf{u},\mathbf{w})\big)$ is achievable by the HK scheme. Also, since $\tilde{U}_{11}=h_{11}U_1\sim \mathsf{p}(\tilde{u}_{11};h_{11}\kappa_1,\theta_1)$ and $\tilde{U}_{21}=h_{21}U_2\in[0,h_{21}\kappa_2]$ with $\mathbb{E}[\tilde{U}_{21}]\le h_{21}\theta_2\kappa_2$, then
		\begin{align} \label{eq:HK-I1}
			\mathsf{R}_1^{\rm u} &= \mathsf{I}(\tilde{U}_{11};\tilde{U}_{11} + \tilde{U}_{21} + Z_1)  \ge \mathsf{F}(0,h_{11}\kappa_1, h_{21}\kappa_2, 1,\theta_1,\theta_2)
			\end{align}
		by Corollary \ref{Cor:FunctionF}. Similarly, we can obtain $\mathsf{R}_2^{\rm u} \ge \mathsf{F}(0,h_{22}\kappa_2, h_{12}\kappa_1, 1,\theta_2,\theta_1)$. Next, observing that $\tilde{W}_{11}=h_{11}W_1\sim \mathsf{p}(\tilde{w}_{11};h_{11}\eta_1,\phi_1)$ and $h_{11}U_1+h_{21}U_2$ is in $[0,h_{11}\kappa_1+h_{21}\kappa_2]$ with $\mathbb{E}[h_{11}U_1+h_{21}U_2]\le h_{11}\theta_1\kappa_1+h_{21}\theta_2\kappa_2=\theta_1'(h_{11}\kappa_1+h_{21}\kappa_2)$, then
		\begin{align}
			\mathsf{R}_1^{\rm w}\ge\mathsf{I}(W_1; Y_1 | W_2) \ge \mathsf{F}(0, h_{11}\eta_1, h_{11}\kappa_1 + h_{21}\kappa_2,1,\phi_1,\theta_1')
			\end{align}
		by Corollary \ref{Cor:FunctionF}. Similarly,
		\begin{align}
			\mathsf{R}_1^{\rm w} &\ge\mathsf{I}(W_1; Y_2 | W_2) \notag\\ 
			&\ge  \mathsf{F}(0, h_{12}\eta_1, h_{12}\kappa_1 + h_{22}\kappa_2,1,\phi_1,\theta_2')\\*
			\mathsf{R}_2^{\rm w} &\ge\mathsf{I}(W_2; Y_1 | W_1) \notag\\ 
			&\ge  \mathsf{F}(0, h_{21}\eta_2, h_{11}\kappa_1 + h_{21}\kappa_2,1,\phi_2,\theta_1')\\*
			\mathsf{R}_2^{\rm w} &\ge\mathsf{I}(W_2; Y_2 | W_1) \notag\\  
			&\ge  \mathsf{F}(0, h_{22}\eta_2, h_{12}\kappa_1 + h_{22}\kappa_2,1,\phi_2,\theta_2').
			\end{align}
		Using similar arguments and using Lemma \ref{Lem:FunctionF}, we can show that
		\begin{align}
			\mathsf{R}_1^{\rm w}+\mathsf{R}_2^{\rm w}&\ge\mathsf{I}(W_1, W_2; Y_i) \notag \\ 
			&\ge  \mathsf{F}( h_{ii}\eta_i, h_{ji}\eta_j, h_{ii}\kappa_i + h_{ji}\kappa_j, \phi_i,\phi_j,\theta_i'),
			\end{align}
		for $i=1,2$. Plugging these lower bounds into \eqref{eq:HK-R} and using Fourier-Motzkin's elimination, we obtain the region $\mathcal{R}'_\text{HK}(\boldsymbol{\kappa},\boldsymbol{\eta},\boldsymbol{\theta},\boldsymbol{\phi})$ described in \eqref{eq:HK-R'}. Since $(\mathsf{R}_1,\mathsf{R}_2)$ are achievable given any $(\boldsymbol{\kappa},\boldsymbol{\eta},\boldsymbol{\theta},\boldsymbol{\phi})$ that satisfies the constraints in Theorem \ref{thm:HK-innerB}, then rate pairs in $\mathcal{R}'_\text{HK}(\boldsymbol{\kappa}, \boldsymbol{\eta}, \boldsymbol{\theta},\boldsymbol{\phi})$ are all achievable. Therefore,  $\mathcal{G}'_\text{HK}$ defined in  \eqref{GprimeHK} is achievable, which ends the proof. 
		\end{proof}
	
	It is worth to note that $\mathcal{G}'_{\rm TIN} \subseteq \mathcal{G}'_{\rm HK}$ since TIN is a special case of the HK scheme when only private messages are sent, i.e., when we fix $\eta_1=\eta_2=0$.

	\subsection{Computable Outer Bounds}
	
	Since providing side information to the receivers does not shrink the capacity region, this provides an outer bound for the capacity of the IC. In this subsection, two outer bounds for the two-user IM/DD IC are obtained based on two choices of side information.	
	
	\subsubsection{Z-Channels}
	The first outer bound is obtained by providing the interference signal to one of the receivers forming a one-sided IC, also known as a Z-channel \cite{costa1985gaussian,kramer2004outer}. Using Fano's inequality, the sum-rate of the Z-channel can be upper-bounded as
	\begin{subequations}
		\begin{align}
			n(\mathsf{R}_1 + \mathsf{R}_2) &\le \mathsf{I}(X_i^n; Y_i^n) + \mathsf{I}(X_j^n; Y_j^n , X_i^n) + n\epsilon_n \\
			& = \mathsf{I}(X_i^n; Y_i^n) + \mathsf{I}(X_j^n; Y_j^n | X_i^n) + n\epsilon_n \label{eq:z_temp} \\
			& = \mathsf{I}(X_1^n, X_2^n; Y_i^n) + \mathsf{I}(X_j^n; Y_j^n | X_i^n) \notag \\
			& \quad - \mathsf{I}(X_j^n; Y_i^n | X_i^n) + n\epsilon_n,
		\end{align}
	\end{subequations}
	for $i,j\in\{1,2\}$ and $i\neq j$, where $\epsilon_n\ge0$ tends to zero as $n\to\infty$ and the second equality follows from the independence of $X_1^n$ and $X_2^n$. Note that \eqref{eq:z_temp} represents a Z-channel where receiver $j$ has no interference from transmitter $i$. 
	This enables us to describe the following outer bound for the IM/DD IC. Given an arbitrary input distribution $p(\mathbf{x}) \in \mathcal{P}(\boldsymbol{\mathsf{A}},\boldsymbol{\alpha})$, an achievable rate pair $(\mathsf{R}_1,\mathsf{R}_2)$ must satisfy, for $i,j\in\{1,2\}$ and $i\ne j$,
	\begin{subequations}
		\label{eq:Z-R}
		\begin{align}
		 \mathsf{R}_i &\le \frac{1}{n}\mathsf{I}(X_i^n; Y_i^n | X_j^n) + \epsilon_n, \label{eq:Z-R-a} \\
		\mathsf{R}_1 + \mathsf{R}_2 &\le \frac{1}{n}\Bigl[\mathsf{I}(X_i^n, X_j^n; Y_i^n) + \mathsf{I}(X_j^n; Y_j^n | X_i^n) \notag \\ &\quad\quad - \mathsf{I}(X_j^n; Y_i^n | X_i^n)\Bigr]+ \epsilon_n.  \label{eq:Z-R-b}
		\end{align}
	\end{subequations}
	The following theorem provides an easily computable outer bound based on \eqref{eq:Z-R}. 
	
	\begin{thm}\label{thm:Z-outerB}
		The capacity region of the IM/DD IC is outer bounded by $\mathcal{G} \subseteq \mathcal{G}'_{\rm Z}$, where $\mathcal{G}'_{\rm Z}$ is the set of all $(\mathsf{R}_1, \mathsf{R}_2)\in\mathbb{R}_+^2$ that satisfy
		\begin{subequations} \label{eq:Z'-R}
			\begin{align} 
			\mathsf{R}_i &\le \overline{\mathsf{C}}(h_{ii}\mathsf{A},\alpha_i), \\ 
			\mathsf{R}_1 + \mathsf{R}_2 &\le \overline{\mathsf{C}}\big(h_{ii}\mathsf{A}+h_{ji}\mathsf{A}, \alpha_i''\big) \notag \\
			&\quad -\min \Biggl\{ 0,  \frac{1}{2} \log \Biggl(  \max\Biggl\{  \frac{h_{ji}^2}{h_{jj}^2},  \frac{1-\frac{h_{ji}^2}{h_{jj}^2}}{2^{2\overline{\mathsf{C}}(h_{jj}\mathsf{A}, \alpha_j)}} \Biggr\} \Biggr) \Biggr\},  
			\end{align}
		\end{subequations}
		where $\alpha_i''=\frac{h_{ii}\alpha_i+h_{ji}\alpha_j}{h_{ii}+h_{ji}}$, $i,j\in\{1,2\}$, and $i\neq j$.
	\end{thm}
	\begin{proof}
		To prove the first inequality above, from \eqref{eq:Z-R-a} we have
		\begin{subequations}
			\label{eq:Ri-Z'}
			\begin{align}
				\mathsf{R}_i &\le \frac{1}{n}  \mathsf{I}(X_i^n;Y_i^n|X_j^n) + \epsilon_n   \\
				&=\frac{1}{n}\mathsf{I}(X_i^n;h_{ii}X_i^n+Z_i^n) + \epsilon_n    \\
				&\le \overline{\mathsf{C}}(h_{ii}\mathsf{A}, \alpha_i) + \epsilon_n, 
			\end{align}
		\end{subequations}
		where the last inequality follows since that, for a memoryless channel described by  $Y_i^{(t)}=X_i^{(t)}+Z_i^{(t)}$, $\mathsf{I}(X_i^n;Y_i^n)\le\sum_{t=1}^{n}\mathsf{I}(X_i^{(t)};Y_i^{(t)})$.
		
		The sum-rate outer bounds in \eqref{eq:Z'-R} is obtained from \eqref{eq:Z-R-b}. Firstly, the first term in \eqref{eq:Z-R-b} satisfies 
		\begin{equation}\label{eq:z'-temp-1}
			\mathsf{I}(X_i^n, X_j^n; Y_i^n) \le n\overline{\mathsf{C}}\bigl((h_{ii}+h_{ji})\mathsf{A}, \alpha_i''\bigr),
		\end{equation}
		which is obtained similar to \eqref{eq:Ri-Z'} above. Secondly, for the remaining two terms in \eqref{eq:Z-R-b},  when $h_{ji}\ge h_{jj}$, we have 
		\begin{subequations}
			\begin{align}
				&\quad\; \mathsf{I}(X_j^n; Y_j^n | X_i^n) - \mathsf{I}(X_j^n; Y_i^n | X_i^n) \notag \\  
				& = \mathsf{I}(X_j^n; h_{jj}X_j^n+Z^n) - \mathsf{I}(X_j^n; h_{ji}X_j^n+Z^n) \\
				& =  \mathsf{I}\Bigl(X_j^n; X_j^n+\frac{Z^n}{h_{jj}}\Bigr) - \mathsf{I}\Bigl(X_j^n; X_j^n+\frac{Z^n}{h_{ji}}\Bigr) \\
				& \le 0,\label{eq:z'-temp-2}
			\end{align}
		\end{subequations}
		where $\mathsf{I}\bigl(X_j^n; X_j^n+\frac{Z^n}{h_{jj}}\bigr) \le \mathsf{I}\bigl(X_j^n; X_j^n+\frac{Z^n}{h_{ji}}\bigr)$ follows since the channel with input $X_j^n$ and output $X_j^n+\frac{Z^n}{h_{jj}}$ can be seen as an degraded version of the less noisy one with the same input but with output $X_j^n+\frac{Z^n}{h_{ji}}$. When $h_{ji}< h_{jj}$, we have the following 
		\begin{subequations}
			\begin{align}
				&\quad\; \mathsf{I}(X_j^n; Y_j^n | X_i^n) - \mathsf{I}(X_j^n; Y_i^n | X_i^n) \notag\\ 
				&= \mathsf{h}(Y_j^n|X_i^n) {-}\mathsf{h}(Y_j^n|X_i^n,X_j^n) {-} \mathsf{h}(Y_i^n|X_i^n) {+} \mathsf{h}(Y_i^n|X_i^n,X_j^n) \\
				&\overset{}{=} \mathsf{h}\left(X_j^n+\frac{Z_j^n}{h_{jj}}\right) -\mathsf{h}\left(X_j^n+\frac{Z_j^n}{h_{jj}} + \tilde{Z}^n\right) + n\log\left( \frac{h_{jj}}{h_{ji}} \right) \label{ind:z_temp_a} \\
				&\overset{}{=} {-}\frac{1}{2} \log\Biggl( 1{+} \frac{2^{2n\mathsf{h}(\tilde{Z})}}{2^{2n\mathsf{h}\left(\frac{Z_j}{h_{jj}}\right)}2^{2\mathsf{I}(X_j^n;h_{jj}X_j^n+Z_j^n)}} \Biggr) {+} n\log\left( \frac{h_{jj}}{h_{ji}} \right) \label{ind:z_temp_b} \\
				&\overset{}{\le} -\frac{1}{2} \log\Biggl( \left(\frac{h_{ji}}{h_{jj}}\right)^{2n}+ \biggl(\frac{1-\frac{h_{ji}^2}{h_{jj}^2}}{2^{2\overline{\mathsf{C}}(h_{jj}\mathsf{A}, \alpha_j)}}\biggr)^n \Biggr) \label{ind:z_temp_c} \\
				&\overset{}{\le} -\frac{n}{2} \log\Biggl(\max\Biggl\{  \frac{h_{ji}^2}{h_{jj}^2},  \frac{1-\frac{h_{ji}^2}{h_{jj}^2}}{2^{2\overline{\mathsf{C}}(h_{jj}\mathsf{A}, \alpha_j)}} \Biggr\} \Biggr), \label{eq:Rsum-Z'}
			\end{align}
		\end{subequations}
		where \eqref{ind:z_temp_a} follows by letting $\tilde{Z} \sim \mathcal{N}\bigl(0, (\frac{1}{h_{ji}^2}-\frac{1}{h_{jj}^2})\sigma^2\bigr)$ independent of $Z_j$, \eqref{ind:z_temp_b} follows by lower-bounding $\mathsf{h}\bigl(X_j^n+\frac{Z_j^n}{h_{jj}} + \tilde{Z}^n\bigr)$ with $\mathsf{h}\bigl(X_j^n+\frac{Z_j^n}{h_{jj}}\bigr)$ and $\mathsf{h}(\tilde{Z}^n)$ using the EPI and then substituting $\mathsf{h}\bigl(X_j^n+\frac{Z_j^n}{h_{jj}}\bigr)=\mathsf{I}\bigl(X_j^n;X_j^n+\frac{Z_j^n}{h_{jj}}\bigr) - \mathsf{h}\bigl(\frac{Z_j^n}{h_{jj}}\bigr)$, \eqref{ind:z_temp_c} follows from the entropy of a zero-mean Gaussian noise and since $\mathsf{I}(X_j^n;h_{jj}X_j^n+Z_j^n) \le n\overline{\mathsf{C}}(h_{jj}\mathsf{A}, \alpha_i)$, and \eqref{eq:Rsum-Z'} follows since the logarithm is a monotonically increasing function and both $\frac{h_{ji}^{2}}{h_{jj}^{2}}$ and $1-\frac{h_{ji}^2}{h_{jj}^2}$ are nonnegative when $h_{ji}< h_{jj}$. Note that \eqref{eq:Rsum-Z'} is negative when $h_{ji} \ge h_{jj}$ and positive otherwise, which allows us to combine it together with \eqref{eq:z'-temp-1} and \eqref{eq:z'-temp-2} while letting $n\to\infty$ to obtain the sum-rate bound in \eqref{eq:Z'-R}. This ends the proof. 
	\end{proof}
	
	\subsubsection{Genie-aided Channel}
	The second outer bound is obtained by providing $S_i = h_{ij}X_i + Z_j$ as side information to receiver $i\in\{1,2\}$, $i\neq j$, forming a genie-aided IC similar to \cite{Etkin-Tse-Wang-2008}. Given an input distribution $p(\mathbf{x}) \in \mathcal{P}(\boldsymbol{\mathsf{A}},\boldsymbol{\alpha})$ defined in Def. \ref{Def:TIN_Feasible_Dist}, an achievable $(\mathsf{R}_1,\mathsf{R}_2)$ must satisfy
	\begin{subequations}
		\label{eq:Ge-R}
		\begin{align}
		\mathsf{R}_i &\le \mathsf{I}(X_i; Y_i | X_j), \quad i,j\in\{1,2\},\ i\neq j,\\
		\mathsf{R}_1 + \mathsf{R}_2 &\le \mathsf{I}(X_1; Y_1, S_1) + \mathsf{I}(X_2; Y_2, S_2).
		\end{align}
	\end{subequations}
	The following theorem provides an easily computable outer bound based on $\mathcal{G}_\text{ge}$. 
	
	\begin{thm}\label{thm:Ge-outerB}
		The capacity region of the IM/DD IC is outer bounded by $\mathcal{G} \subseteq \mathcal{G}'_{\rm ge}$, where $\mathcal{G}'_{\rm ge}$ is the set of all rate pairs $(\mathsf{R}_1, \mathsf{R}_2)\in\mathbb{R}_+^2$ that satisfy
		\begin{subequations}\label{eq:Ge'-R}
			\begin{align} 
			\mathsf{R}_i &\le \overline{\mathsf{C}}(h_{ii}\mathsf{A},\alpha_i),\quad i,j\in\{1,2\},\ i\neq j, \\
			\mathsf{R}_1 + \mathsf{R}_2 &\le \sum_{\substack{i,j=1\\ i\neq j}}^2\overline{\mathsf{C}}\Biggl(\frac{h_{12}h_{21}}{\sqrt{h_{ii}^2+h_{ij}^2}}\mathsf{A}, \alpha_j\Biggr) + \frac{1}{2}\log
			\biggl(1+\frac{h_{ii}^2}{h_{ij}^2}\biggr),
			\end{align}
		\end{subequations}
	\end{thm}
	\begin{proof}
		The proof of the first inequality in \eqref{eq:Ge'-R} is the same as that of \eqref{eq:Z'-R}. The second inequality is obtained from $\mathsf{R}_1+\mathsf{R}_2 \le \mathsf{I}(X_1; Y_1, S_1) + \mathsf{I}(X_2; Y_2, S_2)$ in \eqref{eq:Ge-R} as follows:
		\begin{subequations}
			\begin{align}
				&\quad\; \sum_{i=1}^2 \mathsf{I}(X_i; Y_i, S_i)	\notag \\
				&= \sum_{i=1}^2 \mathsf{I}(X_i;S_i) + \mathsf{I}(X_i;Y_i{|}S_i)\\
				&\overset{}{=} \sum_{i=1}^2 \mathsf{h}(Y_i|S_i)- \mathsf{h}(Z_i)  \label{ind:ge_temp_a} \\
				&\overset{}{=}\sum_{\substack{i,j=1,\\ i\neq j}}^2 I\biggl( X_j; \frac{h_{12}h_{21}}{\sqrt{h_{ii}^2+h_{ij}^2}}X_j + Z_j \biggr) + \frac{1}{2}\log\biggl(  1+\frac{h_{ii}^2}{h_{ij}^2} \biggr)   \label{ind:ge_temp_b}\\
				&\overset{}{\le} \sum_{\substack{i,j=1,\\ i\neq j}}^2 \overline{\mathsf{C}} \Biggl( \frac{h_{12}h_{21}}{\sqrt{h_{ii}^2+h_{ij}^2}}\mathsf{A}, \alpha \Biggr)  + \frac{1}{2}\log \biggl( 1+\frac{h_{ii}^2}{h_{ij}^2} \biggr),
			\end{align}
		\end{subequations}
		where \eqref{ind:ge_temp_a} follows since $\mathsf{h}(S_i|X_i)=\mathsf{h}(Z_j)$ and $\mathsf{h}(Y_i|X_i,S_i)=\mathsf{h}(S_j)$, and \eqref{ind:ge_temp_b} follows from  $\mathsf{h}(Y_i|S_i)=\mathsf{h}(Y_i-\frac{S_i}{h_{ij}}|S_i)$ and dropping the condition, then using the properties $\mathsf{h}(aX)=\mathsf{h}(X) + \log(a)$ for $a>0$ and $\mathsf{h}(X)=\mathsf{I}(X;X+Z)+\mathsf{h}(Z)$. This ends the proof.
	\end{proof}

	\subsection{Numerical Results on the Obtained Capacity Bounds}
	
	We now compare the obtained inner and outer bounds numerically. In Fig. \ref{fig:bounds}, we plot the bounds under $\frac{\mathsf{A}}{\sigma}=10^3$, $\alpha_1=\alpha_2=\alpha = 0.4$ as an example. Let $h_{11}=h_{22}=1$ and vary $h_{12}=h_{21}=g$ within $\{0.01, 0.6, 1, 10\}$. It can be seen that when interference is very weak ($g=0.01$), the TIN and HK schemes achieve almost the same rate region. But when interference becomes larger ($g=0.6$, $1$, or $10$), the achievable region of the HK scheme is strictly larger than the one achieved by TIN. Moreover, it is worth to note that a larger $g$ enables the HK scheme to achieve a larger sum-rate and smaller gap to the upper bound $\mathcal{G}'_{\rm Z}$, and by comparing $g=0.6,1,10$, it can be seen that a stronger interference can be less `harmful' than moderate interference. As the obtained capacity bounds can be easily computed, they can serve as simple performance benchmarks.

	\begin{figure*}
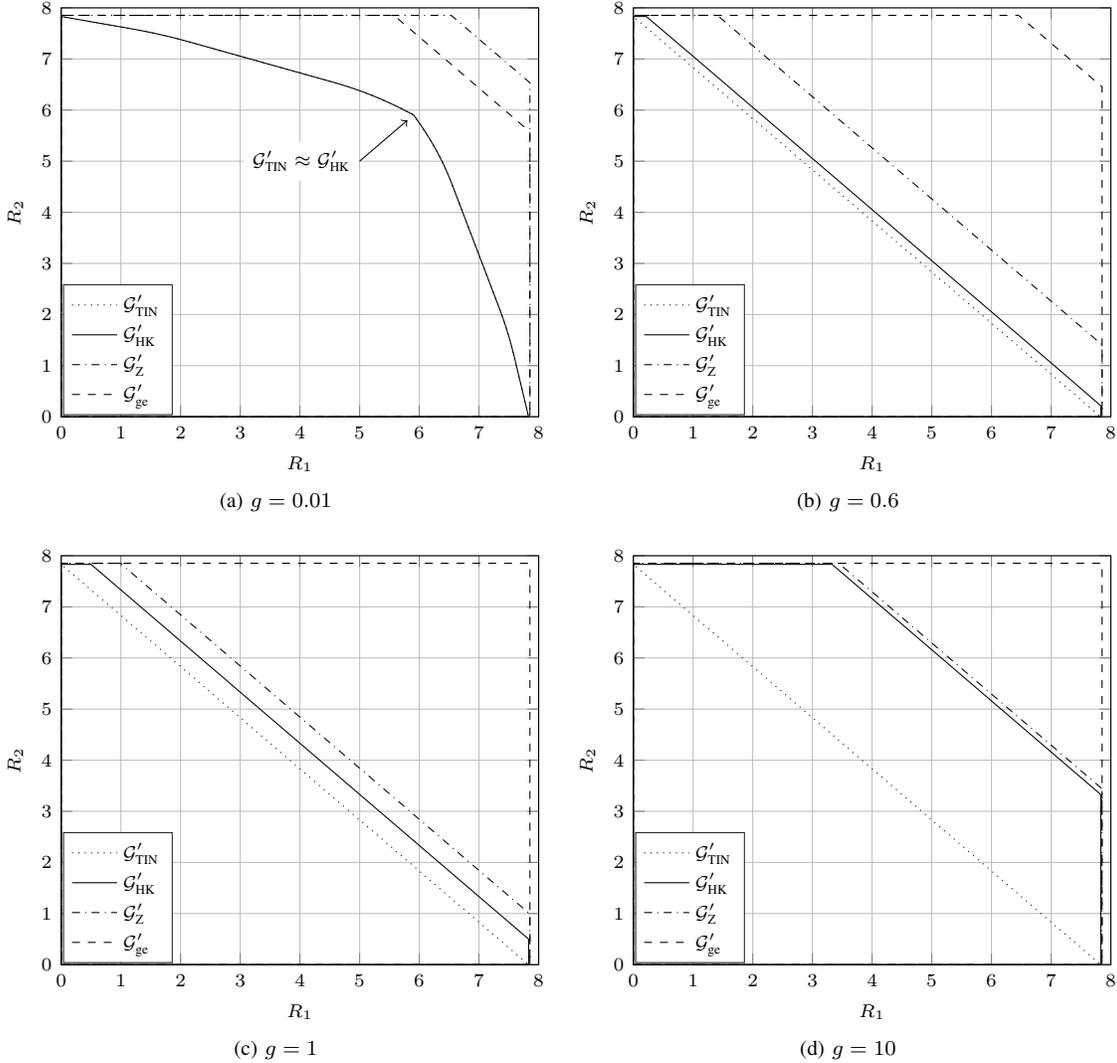

		\centering
		\subfloat[$g=0.01$]{
			\centering
			\input{fig_cmpG_A1e3_g_0p01}
		}
		\subfloat[$g=0.6$]{
			\centering
			\input{fig_cmpG_A1e3_g_0p6}
		}
		
		\subfloat[$g=1$]{
			\centering
			\input{fig_cmpG_A1e3_g_1}
			\label{fig:bounds_g1}
		}
		\subfloat[$g=10$]{
			\centering
			\input{fig_cmpG_A1e3_g_10}
			\label{fig:bounds_g10}
		}
		\caption{Comparison of the obtained inner and outer bounds with $\frac{\mathsf{A}}{\sigma}=10^3$ and $\alpha = 0.4$.}
		\label{fig:bounds}
	\end{figure*}

	In addition to the rate regions in Fig. \ref{fig:bounds}, we plot the quantity $\gamma'(\delta, \mathsf{A}, \alpha)\triangleq\frac{\mathsf{R}_{\rm sum}}{\log\left( 1 + \frac{\rho(\alpha)\mathsf{A}^2}{\sigma^2} \right)}$ versus $\delta=\frac{\log(\frac{g\mathsf{A}}{\sigma})}{\log(\frac{\mathsf{A}}{\sigma})}$ in Fig. \ref{fig:dof-sim}, where $\frac{\mathsf{A}}{\sigma}=10^5$, $\alpha=0.4$, and $\mathsf{R}_{\rm sum}$ is the maximum sum-rate corresponding to each of the inner/outer bounds $\mathcal{G}'_{\rm TIN}$, $\mathcal{G}'_{\rm HK}$, $\mathcal{G}'_{\rm Z}$ and $\mathcal{G}'_{\rm ge}$. We shall see later that the quantity  $\gamma'(\delta, \mathsf{A}, \alpha)$ is related to the GDoF of the IM/DD IC \eqref{eq:dof}. 
	Denote by $\gamma'_{\rm TIN}$, $\gamma'_{\rm HK}$, $\gamma'_{\rm Z}$ and $\gamma'_{\rm ge}$ the $\gamma'(\delta, \mathsf{A}, \alpha)$ value calculated by using $\mathcal{G}'_{\rm TIN}$, $\mathcal{G}'_{\rm HK}$, $\mathcal{G}'_{\rm Z}$ and $\mathcal{G}'_{\rm ge}$, respectively. 
	Overall, these quantities outline a `W' shape in Fig. \ref{fig:dof-sim} between $\max\{\gamma'_{\rm TIN},\gamma'_{\rm HK}\}$ and $\min\{\gamma'_{\rm Z},\gamma'_{\rm ge}\}$, similar to the GDoF as we shall see later (Fig. \ref{fig:dof}). Besides, It can be seen that $\gamma'_{\rm ge}<\gamma'_{\rm Z}$ when $\delta\in(0,0.67)$, and $\gamma'_{\rm Z}<\gamma'_{\rm ge}$ when $\delta\in(0.67,2)$. It can also be seen that $\gamma'_{\rm HK}\ge\gamma'_{\rm TIN}$ for all values of $\delta$ with $\gamma'_{\rm HK}=\gamma'_{\rm TIN}$ when $\delta\in(0,0.5)$, and the largest gap between $\gamma'_{\rm HK}$ and $\min\{\gamma'_{\rm Z},\gamma'_{\rm ge}\}$ occurs when $\delta\in(0.5,1)$, while $\gamma'_{\rm HK}\approx\gamma'_{\rm Z}$ when $\delta\ge1$. 
	
	To compare these results with the rates achieved by using PAM as in \cite[Proposition 5]{dytso2016interference}, we plot $\gamma'_{\rm PAM, \mathsf{m}}$ that is achieved by letting transmitter $i\in\{1,2\}$ only transmit a uniform $N_i$-PAM signal from $[0, 2\alpha\mathsf{A}]$. To evaluate the achievable rate, we fix $N_1=N_2=N$ that is optimally chosen through exhaustively searching over $\{2,3,\dots,10^5\}$, we use no time sharing, and we use \cite[Proposition 2 \& 3]{dytso2016interference} with outage measure (denoted here by $\mathsf{m}$) of $0.1$, $1$, and $10$, respectively, to bound the equivalent constellation seen at the receiver (cf. details in \cite[Sec. III]{dytso2016interference}). 
	It can be seen that in the  weak interference regime\footnote{The weak, moderate, and strong interference regimes will be defined later in Sec. \ref{sec:gdof} through GDoF.} ($\delta\in(0.5,1)$), $\gamma'_{\rm PAM, \mathsf{m}}$ is close to $\gamma'_{\rm ge}$ for all the tested $\mathsf{m}$, which shows the tightness of $\gamma'_{\rm ge}$. But $\gamma'_{\rm PAM, \mathsf{m}}$ always has a much larger gap to the upper bound than the obtained lower bound $\max\{\gamma'_{\rm TIN},\gamma'_{\rm HK}\}$ in the moderate and strong interference regimes. One cause of the rate loss is that the peak intensity of the PAM signal is only up to $2\alpha\mathsf{A}$. Although the capacity bounds of the Gaussian IC from \cite[Proposition 5]{dytso2016interference} are numerically shown to be less tight in the moderate and strong interference regime, it highlights the potential of a practical design that only transmits PAM signals with no time sharing for the IM/DD IC in the weak interference regime. 
	
	\begin{figure*}[!tbp]
		\centering
		\subfloat[$\gamma'(\delta, 10^5, 0.4)$]{
			\centering
			\input{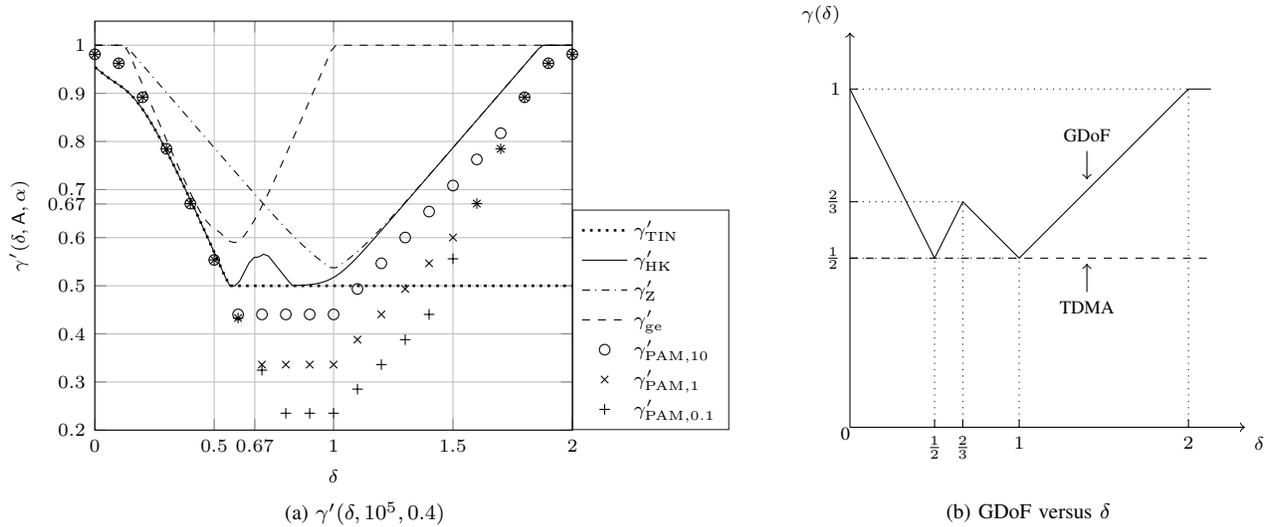}
			\label{fig:dof-sim}
		}\hspace{1 em}
		\subfloat[GDoF versus $\delta$]{
		\centering\raisebox{2ex}[0pt][0pt]{
		\begin{tikzpicture}[scale=1.5]
		\draw[->] (0,0) -- (3.5,0) node[anchor=north west] {\scriptsize $\delta$};
		\draw[->] (0,0) -- (0,3.5) node[anchor=south east] {\scriptsize $\gamma(\delta)$};
		\node[anchor = north east, inner sep=0, outer sep=0] at (0,0) {\scriptsize 0};
		\draw[-] (0,3)--(0.75,1.5)--(1,2)--(1.5,1.5)--(3,3)--(3.2,3);
		
		\draw (.75,2pt) -- (.75,0) node[anchor=north] {\scriptsize $\frac{1}{2}$};
		\draw (1,2pt) -- (1,0) node[anchor=north] {\scriptsize $\frac{2}{3}$};
		\draw (1.5,2pt) -- (1.5,0) node[anchor=north] {\scriptsize $1$};
		\draw (3,2pt) -- (3,0) node[anchor=north] {\scriptsize $2$};
		\draw (2pt,1.5) -- (0,1.5) node[anchor=east] {\scriptsize $\frac{1}{2}$};
		\draw (2pt,2) -- (0,2) node[anchor=east] {\scriptsize $\frac{2}{3}$};
		\draw (2pt,3) -- (0,3) node[anchor=east] {\scriptsize $1$};
		
		\draw[dashed] (0,1.5)--(3.2,1.5);
		\draw[->] (2.1, 1.2) node[anchor=north] {\scriptsize TDMA} -- (2.1, 1.45) ;
		
		\draw[->] (2.1, 2.45) node[anchor=south] {\scriptsize GDoF} -- (2.1, 2.2) ;		
		
		\draw[dotted] (0,3)--(3,3); 
		\draw[dotted] (0,2)--(1,2);
		\draw[dotted] (0,1.5)--(1.5,1.5);
		\draw[dotted] (.75,0)--(.75,1.5);
		\draw[dotted] (1,0)--(1,2);
		\draw[dotted] (1.5,0)--(1.5,1.5);
		\draw[dotted] (3,0)--(3,3);
		\end{tikzpicture}
		\label{fig:dof}
		}}
		\caption{Numerical and theoretical results on the GDoF of the symmetric IM/DD IC. }
		\label{fig:dof_figs}
	\end{figure*}

	\section{Asymptotic Analysis}
	\label{Sec:Asymptotics}
	In addition to a numerical comparison, an analytical comparison is of equal importance as it can provide additional insight. As the obtained computable capacity bounds, $\mathcal{G}'_{\rm TIN}$, $\mathcal{G}'_{\rm HK}$, $\mathcal{G}'_{\rm Z}$ and $\mathcal{G}'_{\rm ge}$, are all based on $\overline{\mathsf{C}}(\cdot)$, we can compare these regions asymptotically by relying on the property (i) in Remark \ref{rmk:property}, i.e.,  $\lim_{\frac{\mathsf{A}}{\sigma} \rightarrow \infty}\bigl[\overline{\mathsf{C}}(\mathsf{A},\alpha) - \underline{\mathsf{C}}(\mathsf{A},\alpha)\bigr] = 0$ \cite{Lapidoth-Moser-Wigger-2009}.
	In this section, we rely on property (i) in Remark \ref{rmk:property} in comparing the inner and outer bounds in the strong and the weak interference regimes and to characterize the asymptotic sum-rate capacity when $\frac{\mathsf{A}}{\sigma}\rightarrow\infty$ and the GDoF, where we focus on the case $\alpha_1=\alpha_2$ with general $h_{ij}$, i.e., the two transmitters are identical, as defined in Sec. \ref{sec:model}.
	
	\subsection{Approximate Asymptotic Capacity Region in the Strong Interference Regime}
	
	The strong interference regime is defined as the regime where $\mathsf{I}(X_1; Y_1|X_2) \le \mathsf{I}(X_1; Y_2|X_2)$ and $\mathsf{I}(X_2; Y_2|X_1) \le \mathsf{I}(X_2; Y_1|X_1)$ \cite[Sec. 6.3]{book-NIT}, leading to $\frac{h_{12}}{h_{11}}\ge1$ and $\frac{h_{21}}{h_{22}}\ge1$. The following theorem presents the asymptotic capacity region in the strong interference regime, which shows that when $\frac{\mathsf{A}}{\sigma}$ is sufficiently large, the simplified HK scheme with only common messages can achieve the capacity region within half a bit.
	
	\begin{thm} \label{thm:SI-Cap}
		In the strong interference regime, given $\alpha_1=\alpha_2=\alpha\in(0,1]$ and $(\mathsf{R}_1,\mathsf{R}_2) \in \mathcal{G}'_Z$ and $\epsilon>0$, the rate pair $(\mathsf{R}_1-\frac{1}{2}-\epsilon,\mathsf{R}_2-\frac{1}{2}-\epsilon)$ can be achieved when $\frac{\mathsf{A}}{\sigma}$ is sufficiently large.
	\end{thm}
	\begin{proof}
		We compare the outer bound $\mathcal{G}'_{\rm Z}$ in Theorem \ref{thm:Z-outerB} and the inner bound $\mathcal{G}'_{\rm HK}$ in Theorem \ref{thm:HK-innerB}. Given $\frac{h_{12}}{h_{11}}\ge1$, $\frac{h_{21}}{h_{22}}\ge1$ (strong interference) and $(\mathsf{R}_1,\mathsf{R}_2)\in\mathcal{G}_{\rm Z}'$, then for $\frac{\mathsf{A}}{\sigma}$ sufficiently large and any $\epsilon>0$, $(\mathsf{R}_1-\epsilon,\mathsf{R}_2-\epsilon)$ is in the region described by
		\begin{subequations}
			\label{eq:Z-R-SI}
			\begin{align}
			\mathsf{R}_i &\le \frac{1}{2}\log \left( 1+ \frac{\rho(\alpha)h_{ii}^2\mathsf{A}^2}{\sigma^2} \right) ,\label{eq:Z-R-SI-i} \\
			\mathsf{R}_1 + \mathsf{R}_2 &\le \frac{1}{2}\log \left( 1+ \frac{\rho(\alpha)(h_{ii}+h_{ji})^2\mathsf{A}^2}{\sigma^2} \right), \label{eq:Z-R-SI-sum}
			\end{align}
		\end{subequations}
		where $i,j\in\{1,2\}$ and $i\neq j$, which follows since $\lim_{\frac{\mathsf{A}}{\sigma} \rightarrow \infty}\varepsilon(\mathsf{A},\alpha_i) = 0$ (property (i) in Remark \ref{rmk:property}).
		
		For $\mathcal{G}'_{\rm HK}$, let $\kappa_i=0$, $\eta_i = \mathsf{A}$ ($i=1,2$) (i.e., the transmitter only transmits the common message with the max peak intensity). Then, we obtain the following rate region
		\begin{subequations}
			\begin{align}
			\mathsf{R}_i &\le  \mathsf{F}(0, h_{ii}\mathsf{A}, 0, 1,\phi_i,0 ), \label{eq:HK-R-SI-i}	\\
			\mathsf{R}_1+\mathsf{R}_2 &\le \mathsf{F}( h_{ii}\mathsf{A}, h_{ji}\mathsf{A}, 0, \phi_i,\phi_j,0 ),  \label{eq:HK-R-SI-sum}
			\end{align}
		\end{subequations}
		where $i,j\in\{1,2\}$, $i\neq j$, and $\phi_i\in(0,\alpha]$ since only common messages are transmitted. Then, the minimum gap between \eqref{eq:Z-R-SI-sum} and \eqref{eq:HK-R-SI-sum} can be upper-bounded as follows,
		\begin{subequations}
			\label{eq:SI-temp-1}
			\begin{IEEEeqnarray}{rCl}
			&& \min_{\phi_i, \phi_j} $\Big\{ $\frac{1}{2}\log \left( 1+ \frac{\rho(\alpha)(h_{ii}+h_{ji})^2\mathsf{A}^2}{\sigma^2} \right) \\ 
			&& \IEEEeqnarraymulticol{1}{r}{- \mathsf{F}( h_{ii}\mathsf{A}, h_{ji}\mathsf{A}, 0, \phi_i,\phi_j,0 )} \Big\} \notag \\ 
			&=& \min_{\phi_i, \phi_j}\frac{1}{2}\log \left( \frac{1  + \rho(\alpha)(h_{ii}+h_{ji})^2 \frac{\mathsf{A}^2}{\sigma^2}}{1  + \bigl(\rho(\phi_i)h_{ii}^2+\rho(\phi_j)h_{ji}^2\bigr) \frac{\mathsf{A}^2}{\sigma^2}} \right)  \\
			&=& \frac{1}{2}\log \left( \frac{1  + \rho(\alpha)(h_{ii}+h_{ji})^2 \frac{\mathsf{A}^2}{\sigma^2}}{1  + \rho(\alpha)(h_{ii}^2+h_{ji}^2) \frac{\mathsf{A}^2}{\sigma^2}} \right)  \\
			&<& \frac{1}{2}\log \left( \frac{(h_{ii}+h_{ji})^2}{h_{ii}^2+h_{ji}^2} \right)  \\
			&\le& \frac{1}{2},
			\end{IEEEeqnarray}
		\end{subequations}
		for $i,j\in\{1,2\}$, $i\neq j$, where the solution of the minimization is $\phi_1=\phi_2=\alpha$ since $\rho(\alpha)$ is monotonically non-decreasing (Property (ii) in Remark \ref{rmk:property}). By substituting $\phi_1=\phi_2=\alpha$ into \eqref{eq:HK-R-SI-i}, we obtain that \eqref{eq:Z-R-SI-i} is achievable, according to the definition of the function $\mathsf{F}$ in \eqref{eq:F}.
		Thus, if $(\mathsf{R}_1,\mathsf{R}_2) \in \mathcal{G}'_Z$, then $(\mathsf{R}_1-\frac{1}{2}-\epsilon,\mathsf{R}_2-\frac{1}{2}-\epsilon) \in \mathcal{G}'_{HK}$, i.e., $(\mathsf{R}_1,\mathsf{R}_2)$ is an achievable rate pair. This concludes the proof.
	\end{proof}
	
	Theorem \ref{thm:SI-Cap} is confirmed for the symmetric IC by the numerical results in Fig. \ref{fig:bounds_g1} and \ref{fig:bounds_g10}, where $\mathcal{G}'_{\rm HK}$ and $\mathcal{G}'_{\rm Z}$ are within a gap $<0.5$ bits when $g\ge1$. Note that the gap tends to 0 as interference becomes stronger, which can be proved by increasing $h_{ji}$ while $h_{ii}$ is held constant in \eqref{eq:SI-temp-1} and observing that the gap decreases to zero.
	
	\subsection{Approximate Asymptotic Sum-Rate Capacity of the Symmetric IC with Weak Interference}
	
	As shown in the last subsection, when $\frac{h_{ji}}{h_{jj}}\ge1$, $i,j\in\{1,2\}$, $i\ne j$, i.e., in the strong interference regime, the simplified HK scheme with only common messages can achieve the capacity within a constant gap with a large $\frac{\mathsf{A}}{\sigma}$. In this subsection, we look into the weak interference regime, i.e., $\frac{h_{ji}}{h_{jj}}<1$, $i,j\in\{1,2\}$, $i\ne j$. We will show that transmitting private messages in this regime achieves the sum capacity within a constant gap. Note that, when only private messages are transmitted, the HK scheme is equivalent to the TIN scheme. 
	
	We have the following avhievable sum-rate through TIN:
	\begin{align}
		\mathsf{R}_{\rm sum, TIN} &= \max_{\substack{ \kappa_i\in[0,\mathsf{A}],\\\theta_i\kappa_i\in[0,\alpha_i\mathsf{A}],\\ i=1,2}}\sum_{\substack{i,j=1\\ i\neq j}}^2\mathsf{F}(0,h_{ii}\kappa_i, h_{ji}\kappa_j, 1,\theta_i,\theta_j) 
	\end{align} 
	Then, using the asymptotic behavior of function $\mathsf{F}$ shown in \eqref{F_Asymptotic}, we have that, for any $\theta_i$ such that $\theta_i\kappa_i\in[0,\alpha_i\mathsf{A}]$ and $\kappa_i=\lambda_i\mathsf{A}$ for all $\lambda_i\in(0,1]$, $i=1,2$,
	\begin{multline}
		\lim_{\frac{\mathsf{A}}{\sigma}\to\infty} \Biggl[ \sum_{\substack{i,j=1\\ i\neq j}}^2\mathsf{F}(0,h_{ii}\kappa_i, h_{ji}\kappa_j, 1,\theta_i,\theta_j) \\ - \frac{1}{2}\log\biggl(1+\frac{\rho(\theta_i)h_{ii}^2\kappa_i^2}{\rho(\theta_j)h_{ji}^2\kappa_j^2} \biggr) \Biggr] = 0. 
	\end{multline}
	If we choose $\lambda_1$, $\lambda_2$, $\theta_1$, and $\theta_2$ such that $\frac{\rho(\theta_1)\kappa_1^2}{\rho(\theta_2)\kappa_2^2}=\frac{h_{21}^2}{h_{12}^2}$, which are guaranteed to exist since $\frac{\rho(\theta_1)\lambda_1}{\rho(\theta_2)\lambda_2}\in[0,\infty)$, then
	\begin{multline}
		\lim_{\frac{\mathsf{A}}{\sigma}\to\infty} \Biggl[\sum_{\substack{i,j=1\\ i\neq j}}^2\mathsf{F}(0,h_{ii}\kappa_i, h_{ji}\kappa_j, 1,\theta_i,\theta_j) \\ - \frac{1}{2}\log\biggl(1+\frac{h_{ii}^2}{h_{ij}^2} \biggr)\Biggr] = 0,
	\end{multline}
	which shows that
	\begin{equation}\label{eq:TIN_sum_weak}
		\lim_{\frac{\mathsf{A}}{\sigma}\to\infty} \Biggl[\mathsf{R}_{\rm sum, TIN} - \log\biggl(1+\frac{h_{ii}^2}{h_{ij}^2} \biggr)\Biggr] \ge 0.
	\end{equation}
	On the other hand, from $\mathcal{G}_{\rm ge}'$ in \eqref{eq:Ge'-R}, we have 
	\begin{align}
		\mathsf{C}_{\rm sum} &\le \sum_{\substack{i,j=1\\ i\neq j}}^2\overline{\mathsf{C}}\biggl(\frac{h_{12}h_{21}}{\sqrt{h_{ii}^2+h_{ij}^2}}\mathsf{A}, \alpha_j\biggr) + \frac{1}{2}\log
		\biggl(1+\frac{h_{ii}^2}{h_{ij}^2}\biggr).  \label{eq:Ge_sum_weak}
	\end{align}
	By comparing \eqref{eq:Ge_sum_weak} with \eqref{eq:TIN_sum_weak}, the gap between the achievable sum-rate $\mathsf{R}_{\rm sum, TIN}$ and the upper bound in \eqref{eq:Ge_sum_weak} satisfies
	\begin{equation}\label{eq:gap_sum_rate_weak}
		\lim_{\frac{\mathsf{A}}{\sigma}\to\infty} \Biggl[\mathsf{C}_{\rm sum} - \mathsf{R}_{\rm sum, TIN} - \sum_{\substack{i,j=1\\ i\neq j}}^2\overline{\mathsf{C}}\biggl(\frac{h_{12}h_{21}}{\sqrt{h_{ii}^2+h_{ij}^2}}\mathsf{A}, \alpha_j\biggr)\Biggr] \le 0. 
	\end{equation}
	The quantity $\frac{h_{12}h_{21}}{\sqrt{h_{ii}^2+h_{ij}^2}}\mathsf{A}$ determines how large the gap shown in \eqref{eq:gap_sum_rate_weak} between $\mathsf{R}_{\rm sum, TIN}$ and $\mathsf{C}_{\rm sum}$ is. We can expect that under some condition, the gap can be bounded. The following theorem uses this result to bound the gap for the symmetric IM/DD IC in the weak interference regime, $h_{12}=h_{21}\in\bigl(0,(\frac{\mathsf{A}}{\sigma})^{-\frac{1}{2}}\bigr)$, which will be justified later in Remark \ref{rmk:regimes} through the GDoF of the symmetric IC.

	\begin{thm}
		Define $\mathsf{G}_{\nicefrac{\mathsf{A}}{\sigma}}$ to be a function of $\frac{\mathsf{A}}{\sigma}$ such that $\mathsf{G}_{\nicefrac{\mathsf{A}}{\sigma}}\in\bigl(0,(\frac{\mathsf{A}}{\sigma})^{-\frac{1}{2}}\bigr)$ and $\lim_{\frac{\mathsf{A}}{\sigma}\to\infty}\frac{\mathsf{G}_{\nicefrac{\mathsf{A}}{\sigma}}\mathsf{A}}{\sigma}=\infty$. For a symmetric IM/DD IC where $\alpha_1=\alpha_2=\alpha\in(0,1]$, $h_{11}=h_{22}=1$, and $h_{12}=h_{21}=\mathsf{G}_{\nicefrac{\mathsf{A}}{\sigma}}$, the asymptotic achievable rate of TIN is within a gap $\lessapprox 1.12$ bits/transmission of the sum capacity, i.e., $\lim_{\frac{\mathsf{A}}{\sigma}\to\infty}\bigl[\mathsf{C}_{\rm sum} - \mathsf{R}_{\rm sum, TIN}\bigr]\le2\max_{x,\alpha\in(0,1]}\overline{\mathsf{C}}(x, \alpha)\lessapprox1.12$ bits/transmission.
	\end{thm}
	\begin{proof}		
		This result can be obtained by substituting $h_{11}=h_{22}=1$, and $h_{12}=h_{21}=\mathsf{G}_{\nicefrac{\mathsf{A}}{\sigma}}$, into \eqref{eq:gap_sum_rate_weak} and noting that $\frac{\mathsf{G}_{\nicefrac{\mathsf{A}}{\sigma}}^2\mathsf{A}}{\sqrt{\mathsf{G}_{\nicefrac{\mathsf{A}}{\sigma}}^2+1}\sigma}\in(0,1]$ when $\mathsf{G}_{\nicefrac{\mathsf{A}}{\sigma}}\in\bigl(0,(\frac{\mathsf{A}}{\sigma})^{-\frac{1}{2}}\bigr)$. Then, using numerical evaluation we obtain $\max_{x,\alpha\in(0,1]}\overline{\mathsf{C}}(x, \alpha)=\overline{\mathsf{C}}(1, 0.5)\lessapprox0.56$ bits/transmission. This ends the proof of the second result.
	\end{proof}

	After seeing the constant-gap tightness in the strong and the weak interference regime, in the next subsection, we show that our obtained capacity inner and outer bounds are tight in the sense of the sum-rate GDoF for the symmetric IC.

	\subsection{Generalized Degrees of Freedom (GDoF)} \label{sec:gdof}
	
	Consider the symmetric IM/DD IC, $\alpha_1=\alpha_2=\alpha$, $h_{11}=h_{22}=1$, and $h_{12}=h_{21}=g$. To obtain the GDoF, we define an interference strength parameter $\delta=\frac{\log\left(\frac{g\mathsf{A}}{\sigma}\right)}{\log\left(\frac{\mathsf{A}}{\sigma}\right)}$ which reflects the strength of the interference-to-noise ratio (INR) relative to the signal-to-noise ratio (SNR). We focus on $\delta>0$, i.e., $g\mathsf{A}>\sigma$, which can be interpreted as the interference-limited regime. This leads to $g=(\frac{\mathsf{A}}{\sigma})^{\delta-1}$ and $\lim_{\frac{\mathsf{A}}{\sigma}\to\infty}\frac{g\mathsf{A}}{\sigma}=\infty$. Then, the GDoF of the symmetric IM/DD IC can be defined as \cite[eq. (20)]{Etkin-Tse-Wang-2008}
	\begin{align}\label{eq:dof}
	\gamma (\delta) &\triangleq \lim_{  \substack{ \frac{\mathsf{A}}{\sigma} \rightarrow\infty } } \frac{  \mathsf{C}_{\rm sum} }{ 2\mathsf{C}(\mathsf{A},\alpha) } =\lim_{\substack{\frac{\mathsf{A}}{\sigma}\rightarrow\infty}} \frac{  \mathsf{C}_{\rm sum} }{ 2\log\left( \frac{\mathsf{A}}{\sigma}  \right) },
	\end{align}
	where the last equality holds for all $\alpha\in(0,1]$ since $\lim_{\frac{\mathsf{A}}{\sigma}\to\infty} \bigl[\mathsf{C}(\mathsf{A},\alpha)-\underline{\mathsf{C}}(\mathsf{A},\alpha)\bigr]=0$ from Property (i) in Remark \ref{rmk:property}. It is worth to mention that the GDoF can be interpreted as the prelog factor to approximate the sum-rate capacity for various interference strengths at high $\frac{\mathsf{A}}{\sigma}$.
	The theorem below characterizes the GDoF of the symmetric IM/DD IC, which is depicted and verified in Fig. \ref{fig:dof}.
	
	\begin{thm} \label{thm:dof}
		The GDoF of the symmetric IM/DD IC is 
		\begin{equation}
				\gamma(\delta) = \min\Bigl\{\max\bigl\{1- \delta, \delta\bigr\},\max\bigl\{1- \frac{\delta}{2}, \frac{\delta}{2} \bigr\}, 1\Bigr\}.
		\end{equation}
	\end{thm}	
	\begin{proof}
		We give an outline of the proof here, and relegate details to the following subsections. First, to derive a GDoF upper bound, we use the Z-channel bound in Theorem \ref{thm:Z-outerB} to show that $\gamma(\delta) \le \min\left\{\max \left\{ \frac{\delta}{2} , 1-\frac{\delta}{2}  \right\}, 1 \right\}$ as detailed in Sec. \ref{sec:gdof-Z}, and the genie-aided bound in Theorem \ref{thm:Ge-outerB} to show that $\gamma(\delta)\leq\max\left\{1-\delta,\delta\right\}$ as detailed in Sec. \ref{sec:gdof-ge}. Then, we show that the HK scheme achieves the obtained GDoF upper bound in Sec. \ref{sec:gdof-hk}, which proves the theorem. 
	\end{proof}
	
	\begin{remark} \label{rmk:regimes}
		The GDoF given in Theorem \ref{thm:dof} helps to define the weak, moderate, and strong interference regimes for the symmetric IM/DD IC as the regimes where $g\in\left(0, (\frac{\mathsf{A}}{\sigma})^{-\frac{1}{2}} \right)$, $\left[ (\frac{\mathsf{A}}{\sigma})^{-\frac{1}{2}}, 1\right)$, and $[1,\infty)$, respectively. Compared to the GDoF of the standard Gaussian IC \cite[(25)]{Etkin-Tse-Wang-2008}, apart from the different definitions of $\text{SNR}$ and $\text{INR}$ (defined here as $\frac{\mathsf{A}}{\sigma}$ and $\frac{g\mathsf{A}}{\sigma}$, respectively), we can see that the GDoF of the symmetric IM/DD IC and the symmetric Gaussian IC are the same.
\end{remark}
	
	\begin{remark}
		Although the derivation of GDoF $\gamma(\delta)$ requires $g$ to be dependent on $\frac{\mathsf{A}}{\sigma}$, the approximation $\mathsf{C}_{\rm sum}\approx2\gamma(\delta)\log(\frac{\mathsf{A}}{\sigma})$ applies for any $g>0$ when $\frac{\mathsf{A}}{\sigma}$ is sufficiently large and $\delta$ is calculated as $\frac{\log(\frac{g\mathsf{A}}{\sigma})}{\log(\frac{\mathsf{A}}{\sigma})}$, whereby the approximation becomes more precise as $\frac{\mathsf{A}}{\sigma}$ grows. Fig. \ref{fig:dof_figs} shows an example.
	\end{remark}

	Next, we will derive the upper and lower bounds of $\gamma(\delta)$ through $\mathcal{G}'_{\rm Z}$, $\mathcal{G}'_{\rm ge}$, $\mathcal{G}'_{\rm TIN}$, and $\mathcal{G}'_{\rm HK}$, where the fact that 
	\begin{align}\label{Asymp_Cap}
	\lim_{\frac{\mathsf{A}}{\sigma}\to\infty}\frac{\overline{\mathsf{C}}(\mathsf{A},\alpha)}{\log\left(\frac{\mathsf{A}}{\sigma}\right)}=\lim_{\frac{\mathsf{A}}{\sigma}\to\infty}\frac{\underline{\mathsf{C}}(\mathsf{A},\alpha)}{\log\left(\frac{\mathsf{A}}{\sigma}\right)}=1,
	\end{align}
	 the asymptotic behavior of the function $\mathsf{F}$ shown in \eqref{F_Asymptotic}, and properties (i)-(iii) listed in Sec. \ref{sec:p2p-bounds} will be extensively used.
	
	\subsubsection{Z-Channel GDoF Upper Bound} \label{sec:gdof-Z}
	Recall $\mathcal{G}'_{\rm Z}$ \eqref{eq:Z'-R} in Theorem \ref{thm:Z-outerB}, from which we have
	\begin{align}
	\mathsf{C}_{\rm sum} &\le \max_{(\mathsf{R}_1,\mathsf{R}_2) \in \mathcal{G}'_{\rm Z}} \mathsf{R}_1 + \mathsf{R}_2 \notag \\ 
	& \le \overline{\mathsf{C}}\bigl((1+g)\mathsf{A},\alpha\bigr) \notag \\
	&\quad - \min \Biggl\{ 0,  \frac{1}{2} \log \Biggl(  \max\biggl\{ g^2, \frac{1-g^2}{2^{ 2\overline{\mathsf{C}}(\mathsf{A},\alpha) }} \biggr\} \Biggr) \Biggr\}.  
	\end{align}
	
	When $\delta\ge1$ (i.e., $g\ge 1$), $\mathsf{C}_{\rm sum}\le \overline{\mathsf{C}}((1+g)\mathsf{A},\alpha)$, and we have
	\begin{align} 
	\gamma(\delta)
	\le \lim_{\frac{\mathsf{A}}{\sigma}\to\infty} \frac{ \overline{\mathsf{C}}((1+g)\mathsf{A},\alpha)}{ 2\log \bigl( \frac{ \mathsf{A}}{\sigma} \bigr)  } =\lim_{\frac{\mathsf{A}}{\sigma}\to\infty} \frac{ \log \bigl( \frac{(g+1)\mathsf{A}}{\sigma} \bigr)  }{ 2\log \bigl( \frac{ \mathsf{A}}{\sigma} \bigr)  } = \frac{\delta}{2},
	\end{align}
	which follows using \eqref{Asymp_Cap} and since $g=(\frac{\mathsf{A}}{\sigma})^{\delta-1}$.
	On the other hand, when $\delta<1$ (i.e., $g< 1$), we have $\mathsf{C}_{\rm sum} \le \overline{\mathsf{C}}\bigl((1+g)\mathsf{A},\alpha\bigr)-\log(g)$ since $\max\left\{ x, y \right\}\ge x$, and
	\begin{subequations}
		\begin{align}
			\gamma(\delta)&\le \lim_{\frac{\mathsf{A}}{\sigma} \rightarrow \infty}\frac{\overline{\mathsf{C}}\bigl((1+g)\mathsf{A},\alpha\bigr)-\log(g)}{2\log\bigl(\frac{\mathsf{A}}{\sigma}\bigr)} \\
			&= \lim_{\frac{\mathsf{A}}{\sigma} \rightarrow \infty}\frac{\log\bigl(\frac{(1+g)\mathsf{A}}{\sigma}\bigr)-\log(g)}{2\log\bigl(\frac{\mathsf{A}}{\sigma}\bigr)} \\
			&= 1- \frac{\delta}{2},
		\end{align}
	\end{subequations}
	which follows using \eqref{Asymp_Cap} and since $g=(\frac{\mathsf{A}}{\sigma})^{\delta-1}$. 
	Thus, we conclude that
	\begin{align}\label{eq:GDOFZ}
	\gamma(\delta) \le \min\biggl\{\max \Bigl\{ \frac{\delta}{2} , 1-\frac{\delta}{2}  \Bigr\}, 1 \biggr\}.
	\end{align}
	
	\subsubsection{Genie-Aided Channel GDoF Upper Bound} \label{sec:gdof-ge}
	From $\mathcal{G}'_{\rm ge}$ in Theorem \ref{thm:Ge-outerB}, we have 
	\begin{subequations}
		\begin{align}
		\mathsf{C}_{\rm sum} & \le \max_{(\mathsf{R}_1,\mathsf{R}_2) \in \mathcal{G}'_{\rm ge}} \mathsf{R}_1 + \mathsf{R}_2 \\
		&\le 2\overline{\mathsf{C}}\Bigl(\frac{g^2}{\sqrt{g^2+1}}\mathsf{A}, \alpha\Bigr) + \log\Bigl(1 + \frac{1}{g^2}\Bigr)  \\
		& = \log\Bigl( 1+\frac{1}{g^2} + \frac{\rho(\alpha)g^2\mathsf{A}^2}{\sigma^2} \Bigr) + 2\max_{x\in\mathbb{R}_+}\varepsilon(x,\alpha),
		\end{align}
	\end{subequations}
	where $\max_{x\in\mathbb{R}_+}\varepsilon(x,\alpha)<0.7$ bits (Property (i) in Remark \ref{rmk:property}). Then, we have
	\begin{subequations}
		\begin{align}
			\gamma(\delta) &\le 
			\lim_{\frac{\mathsf{A}}{\sigma}\to\infty} \frac{\log\Bigl( 1+\frac{1}{g^2} + \frac{\rho(\alpha)g^2\mathsf{A}^2}{\sigma^2} \Bigr) + 2\max_{x\in\mathbb{R}_+}\varepsilon(x,\alpha)}{2\log(\frac{\mathsf{A}}{\sigma})} \\
			&\leq
			\lim_{\frac{\mathsf{A}}{\sigma}\to\infty} \max \biggl\{ -\frac{\log(g)}{\log(\frac{\mathsf{A}}{\sigma})}, \frac{\log(\frac{g\mathsf{A}}{\sigma})}{\log(\frac{\mathsf{A}}{\sigma})} \biggr\} \\
			&\overset{}{=} \max\{1-\delta, \delta\},\label{eq:GDOFGA}
		\end{align}
	\end{subequations}
	which follows since $\log\left( 1+x + y \right) \le \max\bigl\{\log(1+ 2x), \log(1+2y)\bigr\}$ and using and $g=(\frac{\mathsf{A}}{\sigma})^{\delta-1}$. 
	
	Combining \eqref{eq:GDOFZ} and \eqref{eq:GDOFGA}, we obtain 
	\begin{align}
	\gamma(\delta) \le \min\Bigl\{\max\bigl\{1- \delta, \delta\bigr\}, \max\bigl\{1- \frac{\delta}{2}, \frac{\delta}{2} \bigr\}, 1\Bigr\},
	\end{align} 
	which coincides with the GDoF in Theorem \ref{thm:dof}. Next, we show that this upper bound is achievable.

	\subsubsection{Achievable GDoF of the HK Scheme} \label{sec:gdof-hk}
	We first look at the case when only a private message is used in the HK scheme, which is equivalent to TIN. From Theorem \ref{thm:TIN-innerB} we have
	\begin{align}
		\mathsf{C}_{\rm sum} \ge \max_{(\mathsf{R}_1,\mathsf{R}_2) \in \mathcal{G}'_{\rm TIN}} \mathsf{R}_1 + \mathsf{R}_2  
		\ge 2\mathsf{F}(0,\mathsf{A}, g\mathsf{A}, 1,\alpha,\alpha),  
	\end{align}
	which follows by dropping time-sharing and setting $\kappa_1=\kappa_2=\mathsf{A}$ and $\theta_1=\theta_2=\alpha$. 
	Moreover, as $(\mathsf{R}_1,\mathsf{R}_2)=\bigl(\mathsf{C}(\mathsf{A},\alpha),0\bigr)$ and $\bigl(0,\mathsf{C}(\mathsf{A},\alpha)\bigr)$ are achievable, $\mathsf{C}_{\rm sum} \ge \mathsf{C}(\mathsf{A},\alpha)$ is achievable (TDMA). Thus, $\mathsf{C}_{\rm sum} \ge \max\bigl\{2\mathsf{F}(0,\mathsf{A}, g\mathsf{A}, 1,\alpha,\alpha), \mathsf{C}(\mathsf{A},\alpha)\bigr\}$. Then, 
	\begin{subequations}
		\begin{align}
			\gamma(\delta) &=  \lim_{\frac{\mathsf{A}}{\sigma}\to\infty} \frac{\mathsf{C}_{\rm sum}}{2\log\left( \frac{\mathsf{A}}{\sigma}  \right)} 
			\\
			&\ge \lim_{\frac{\mathsf{A}}{\sigma}\to\infty} \frac{\max \left\{ \log\bigl( 1+\frac{1}{g^2} \bigr),\log(\frac{\mathsf{A}}{\sigma}) \right\}}{2\log\left( \frac{\mathsf{A}}{\sigma}  \right)}  \\
			&=  \max \Bigl\{ 1- \delta, \frac{1}{2}\Bigr\},
		\end{align}
	\end{subequations}
	where we used \eqref{F_Asymptotic} and $g=(\frac{\mathsf{A}}{\sigma})^{\delta-1}$.
	Comparing to the obtained GDoF upper bound proves that $\gamma(\delta)=1-\delta$ is achievable. 
	
	Secondly, when $\delta\ge1$ (i.e., $g\ge1$), we have that the gap between $\mathcal{G}'_{\rm HK}$ and $\mathcal{G}'_{\rm Z}$ is less than a half bit as $\frac{\mathsf{A}}{\sigma}\rightarrow\infty$ from the proof of Theorem \ref{thm:SI-Cap}. Thus, the GDoF characterized by $\mathcal{G}'_{\rm Z}$ is achievable by $\mathcal{G}'_{\rm HK}$, i.e., $\gamma(\delta)=\min\{\frac{\delta}{2},1\}$ is achievable when $\delta\ge1$.
	
	Lastly, we look into the case where $\delta \in [\frac{1}{2},1)$, i.e., $g\in\bigl[(\frac{\mathsf{A}}{\sigma})^{-\frac{1}{2}},1\bigr)$. By dropping time-sharing and letting $\kappa_i=\beta\mathsf{A}$, $\eta_i=(1-\beta)\mathsf{A}$ for some $\beta\in(0,1)$, and $\theta_i=\phi_i=\alpha$, $i=1,2$,  we have\footnote{$\mathsf{C}_{\rm sum} \ge \min\{ 2\times\eqref{eq:HK-R'-1}, \eqref{eq:HK-R'-2} \}$.}
	\begin{align}
	\mathsf{C}_{\rm sum} \ge \max_{0<\beta\le 1} 2f_1 + \min\{2f_2, f_3\}
	\end{align}
	where 
	\begin{subequations}
		\begin{align}
			f_1 &= \mathsf{F}(0,\beta\mathsf{A}, g\beta\mathsf{A},1,\alpha,\alpha),\\
			f_2 &= \mathsf{F}\big(0,g(1-\beta)\mathsf{A}, (g+1)\beta\mathsf{A}, 1,\alpha,\alpha
			\big), \\
			f_3 &= \mathsf{F}\big((1-\beta)\mathsf{A}, (1-\beta)g\mathsf{A}, (g+1)\beta\mathsf{A}, \alpha,\alpha,\alpha \big).
		\end{align}
	\end{subequations}
	Letting $\beta = c(\frac{g\mathsf{A}}{\sigma})^{-1}$ for any constant $c$, we obtain
	\begin{subequations}
		\begin{align}
			\gamma(\delta)&\ge\lim_{\frac{\mathsf{A}}{\sigma}\to\infty} \frac{2f_1}{2\log(\frac{\mathsf{A}}{\sigma})} + \min \biggl\{\frac{2f_2}{2\log(\frac{\mathsf{A}}{\sigma})}, \frac{f_3}{2\log(\frac{\mathsf{A}}{\sigma})} \biggr\} \\*
			&= \lim_{\frac{\mathsf{A}}{\sigma}\to\infty}  \frac{\log\big(1+\frac{1}{g^2}\big)}{2\log\left( \frac{\mathsf{A}}{\sigma}  \right)} \notag \\ & \hspace{-1cm} +  \min\Bigg\{  \frac{\log\Big(1+\frac{g^4\mathsf{A}^2}{c^2(1+g)^2\sigma^2}\Big)}{2\log\left( \frac{\mathsf{A}}{\sigma}  \right)}, \frac{\frac{1}{2}\log\Big(1+ \frac{(1+g^2)}{c^2(1+g)^2}\frac{g^2\mathsf{A}^2}{\sigma^2}\Big)}{2\log\left( \frac{\mathsf{A}}{\sigma}  \right)}
			\Bigg\} \\*
			&\ge 1-\delta + \min\Bigl\{2\delta-1,\frac{\delta}{2}\Bigr\}\\*
			& = \min\Bigl\{ \delta, 1-\frac{\delta}{2} \Bigr\}.
		\end{align}
	\end{subequations}
	which follows by using \eqref{F_Asymptotic} and since $c$ is a constant.
	This coincides with the GDoF upper bound in the interval $\delta\in\left[\frac{1}{2},1\right)$. Thus, the HK scheme achieves the GDoF over the whole range of $\delta$. This proves the achievability of Theorem \ref{thm:dof} and ends the proof.

\section{Sample Applications of the Interference Channel}
	\label{Sec:Numerical}
	In this section, two example OWC application scenarios are investigated and the obtained results are adopted for solving system design problems. We investigate an on-chip OWC example, where OWC can be used for wireless bus-based interconnections \cite{bermond1996bus,kirman2006leveraging,nafari2017chip} and interference can be moderate/strong. We show how the obtained computable capacity bounds can aid in optimizing the receivers' locations.
	We also investigate an indoor VLC example, where pairing a transmitter with a nearby receiver typically leads to weak interference. We show that the results can be used to compare the interference management capabilities of different schemes.

	\subsection{On-chip Optical Wireless Communication}
	
	Consider a multi-core chip inside a 6$\times$6cm package, where cores want to communicate with each other. Connecting the cores using wires poses a challenge in terms of space. Instead, an IM/DD OWC system with multiple transmitters and receivers can be used, to save space \cite{bermond1996bus,kirman2006leveraging,nafari2017chip}. Assume that there are two identical transmitters-receiver pairs which transmit light signals through a free-space wireless channel 
	which leads to cross talk. The transmitters are at positions $(1.5,0)$ and $(4.5,0)$ cm facing the $(0,1)$ direction where the origin is at the bottom-left corner of the chip as shown in Fig. \ref{fig:onChip}, and the receivers are at positions $(0, 2)$ and $(x,y)$cm facing the $(0,-1)$ direction where $x$ and $y$ are to be determined. Also, assume that all links are line-of-sight dominated. The peak intensity constraint of the transmitters is $\mathsf{A}=1000$, the average intensity constraint is $\alpha\mathsf{A}$ where $\alpha=0.5$, and the noise variance $\sigma^2=1$. Given the above settings, the design task is to find a location for the second receiver to achieve the highest possible sum-rate or meet a minimum sum-rate requirement.
	
	The solution depends on the transmission scheme. Computable achievable rates of the TIN and HK schemes are given in Theorems \ref{thm:TIN-innerB} and \ref{thm:HK-innerB}. Using these Theorems, the achievable sum-rate of each scheme can be evaluated at different values of $(x,y)$,  after evaluating the channel gains $h_{ij}$, $i=1,2$, as follows (assuming the transmitters are modeled as Lambertian emitters)
	\begin{align} \label{eq:lambertian}
	h_{ij} = \frac{(m+1)\mathsf{GS}}{2\pi d_{ij}^2} (\cos \phi_{ij})^m \cos \psi_{ij} \; \mathbb{I}(\psi_{ij} < \Psi),
	\end{align}
	where $m = \frac{-\ln 2}{\ln (\cos \Phi)}$ is the Lambertian emission order, $\Phi$ is the half-intensity semi-angle of the light source, $d_{ij}$ is the transmitter-receiver distance, \cc{$\mathsf{G}$} is the receiver gain, \cc{$\mathsf{S}$} is the area of the receiver's  active region, $\phi_{ij}$ and $\psi_{ij}$ are the emitting and receiving angles, respectively, i.e., the angle between the transmitter (or receiver) direction and the line between the transmitter and receiver, $\Psi$ is the field of view (FOV) of the receiver, and $\mathbb{I}(\cdot)$ is the indicator function which returns 1 if its argument is true, and zero otherwise \cite{ghassemlooy2019optical}. We assume here that \cc{$\mathsf{G}=1$, $\mathsf{S}=10\;\text{mm}^2$}, $\Phi = 60^{\circ}$, $\Psi = 70^{\circ}$, and since the transmitters and receivers are pointed along the y-axis, we have that $\phi_{ij} = \psi_{ij}$ which can be obtained from the transmitter and receiver positions. 
	
	Fig. \ref{fig:onChip_countour} shows a contour plot of the achievable sum-rate (bit/transmission) with respect to the position of receiver 2. It can be seen that the largest sum-rate is achieved around $(x,y)=(4.7,1.3)$ cm for both the TIN and HK schemes, since receiver 2 nearly does not receive any interference at this point, whereas receiver 1 receives strong interference. However, if the receiver is constrained to be in the upper-right quadrant due to space constraints, then the HK scheme provides better performance as demonstrated in Fig. \ref{fig:onChip_G}, which shows the achievable rate regions of both schemes when receiver 2 is at $(4,4)$ cm (where $[h_{11},h_{12},h_{21},h_{22}]=[3.3,1,3.3,1.9]$, correspondingly. Also, if the design targets a minimum sum-rate requirement of $0.4$ bit/transmission e.g., then the HK scheme can provide a wider area to place receiver 2 which is more convenient for the on-chip layout. Note that the sum-rate design target defined here only serves as an example. Various design targets can be analyzed similarly (such as minimum achievable rate per receiver), and precise channel models from measurements can be used instead of \eqref{eq:lambertian} to obtain performance baselines using the TIN and HK schemes.
	\begin{figure}
		\centering
		\subfloat[Sum-rate countour versus Rx2 position]{\label{fig:onChip_countour}
			\scalebox{1}{
%
%

%
\begin{tikzpicture}

\begin{axis}[%
width=.9\columnwidth,
height=.8\columnwidth,
at={(0,0)},
xmin=-3,
xmax=3,
xlabel={\scriptsize X(cm)},
xlabel near ticks,
ymin=0,
ymax=6,
ylabel={\scriptsize Y(cm)},
ylabel near ticks,
axis background/.style={fill=white},
legend style={legend cell align=left, align=left, at={(axis cs: -3,6)}, anchor=north west,},
x tick label style ={font=\scriptsize},
y tick label style ={font=\scriptsize},
xlabel near ticks,
ylabel near ticks,
xtick = {-3,-2,-1,0,1,2,3},
xticklabels={0,1,2,3,4,5,6},
ytick = {0, 1,2,3,4,5,6},
grid=major,
]
\addplot[contour prepared, contour prepared format=matlab, contour/draw color=black] table[row sep=crcr] {%
	0.4	33\\
	3	1.95585030621715\\
	2.88090931872727	2.21052631578947\\
	2.68421052631579	2.44943466801666\\
	2.60965357074291	2.52631578947368\\
	2.36842105263158	2.72216062589535\\
	2.13709023585806	2.84210526315789\\
	2.05263157894737	2.91503706953543\\
	1.73684210526316	3.00089610290489\\
	1.42105263157895	3.01152740040215\\
	1.10526315789474	2.9727472126271\\
	0.871011532473068	2.84210526315789\\
	0.789473684210526	2.81087160100282\\
	0.473684210526316	2.60181019897806\\
	0.398190476340831	2.52631578947368\\
	0.191963332838208	2.21052631578947\\
	0.157894736842105	2.13964673496854\\
	-0.00702391574900915	1.89473684210526\\
	-0.0609585209149037	1.57894736842105\\
	-0.0695041447759698	1.26315789473684\\
	-0.00640150426951192	0.947368421052632\\
	0.157894736842105	0.783069310174568\\
	0.473684210526316	0.656360032109582\\
	0.789473684210526	0.64303878196313\\
	1.09380332330003	0.947368421052632\\
	1.10526315789474	0.957638982999829\\
	1.42105263157895	0.956342348229948\\
	1.73684210526316	0.956747800218549\\
	2.04325219978145	1.26315789473684\\
	2.05263157894737	1.2796320982461\\
	2.36842105263158	1.28775195150519\\
	2.68421052631579	1.31128142815142\\
	2.95187646658542	1.57894736842105\\
	3	1.84490153329724\\
	0.6	29\\
	0.233003648826896	0.947368421052632\\
	0.473684210526316	0.746057749205637\\
	0.789473684210526	0.684518846785976\\
	1.05232325847718	0.947368421052632\\
	1.10526315789474	0.994814351615838\\
	1.42105263157895	0.988824412255763\\
	1.73684210526316	0.99069743994371\\
	2.00930256005629	1.26315789473684\\
	2.05263157894737	1.3392621935866\\
	2.36842105263158	1.37677270137525\\
	2.68421052631579	1.4854695692667\\
	2.77768832547014	1.57894736842105\\
	2.68421052631579	1.85832535758637\\
	2.66581813905751	1.89473684210526\\
	2.42897483652252	2.21052631578947\\
	2.36842105263158	2.27221743432982\\
	2.05263157894737	2.48147008438246\\
	1.90078878659533	2.52631578947368\\
	1.73684210526316	2.58370872474544\\
	1.42105263157895	2.60453716064107\\
	1.10526315789474	2.53691542675501\\
	1.08371723557504	2.52631578947368\\
	0.789473684210526	2.39482630247081\\
	0.581452881708153	2.21052631578947\\
	0.473684210526316	2.09590375859245\\
	0.343962295179093	1.89473684210526\\
	0.228060553005478	1.57894736842105\\
	0.197124084139292	1.26315789473684\\
	0.233003648826896	0.947368421052632\\
	1.2	19\\
	0.531995254480548	0.947368421052632\\
	0.789473684210526	0.808959041254513\\
	0.927883064008645	0.947368421052632\\
	1.10526315789474	1.10634045746387\\
	1.42105263157895	1.08627060433321\\
	1.73684210526316	1.09254635911919\\
	1.90745364088081	1.26315789473684\\
	2.05263157894737	1.51815247960809\\
	2.23677196904835	1.57894736842105\\
	2.05263157894737	1.76322528085358\\
	1.7472799990056	1.89473684210526\\
	1.73684210526316	1.89950609113603\\
	1.42105263157895	1.92885433690145\\
	1.29948029670915	1.89473684210526\\
	1.10526315789474	1.84814645439625\\
	0.789473684210526	1.61646261825438\\
	0.760096896646328	1.57894736842105\\
	0.583142869815884	1.26315789473684\\
	0.531995254480548	0.947368421052632\\
	2.2	7\\
	1.30638146629293	1.26315789473684\\
	1.42105263157895	1.24868092446228\\
	1.73684210526316	1.262294557745\\
	1.737705442255	1.26315789473684\\
	1.73684210526316	1.26599980670413\\
	1.42105263157895	1.31019699754032\\
	1.30638146629293	1.26315789473684\\
};
\addlegendentry{\scriptsize TIN}

\addplot[contour prepared, contour prepared format=matlab, dashed, contour/draw color=red, color=red, line width = 1] table[row sep=crcr] {%
	0.4	48\\
	3	4.34524714780987\\
	2.82655679789733	4.42105263157895\\
	2.68421052631579	4.49014262408112\\
	2.36842105263158	4.58029886025366\\
	2.05263157894737	4.61329670017736\\
	1.73684210526316	4.59694263082326\\
	1.42105263157895	4.52828124199576\\
	1.16480353984802	4.42105263157895\\
	1.10526315789474	4.40024670237569\\
	0.789473684210526	4.22527557213132\\
	0.644011184868583	4.10526315789474\\
	0.473684210526316	3.96516024525141\\
	0.315983920420453	3.78947368421053\\
	0.157894736842105	3.58393145553102\\
	0.087030374348011	3.47368421052632\\
	-0.0743403532333378	3.15789473684211\\
	-0.157894736842105	2.87394623894334\\
	-0.372560775454653	2.84210526315789\\
	-0.473684210526316	2.84101478163276\\
	-0.47684805985173	2.84210526315789\\
	-0.789473684210526	2.9939076008772\\
	-1.10526315789474	2.85285794711685\\
	-1.12806154799794	2.84210526315789\\
	-1.13557751360901	2.52631578947368\\
	-1.10526315789474	2.29348582431706\\
	-1.09920913130878	2.21052631578947\\
	-1.04285300899299	1.89473684210526\\
	-0.938018941470452	1.57894736842105\\
	-0.789473684210526	1.3439645323556\\
	-0.760403665840876	1.26315789473684\\
	-0.54455238730069	0.947368421052632\\
	-0.473684210526316	0.902380089131624\\
	-0.157894736842105	0.701310367999072\\
	-0.0211954513287974	0.631578947368421\\
	0.157894736842105	0.486952873125985\\
	0.302520811084541	0.631578947368421\\
	0.473684210526316	0.654701587446538\\
	0.789473684210526	0.642352033702952\\
	1.09449007156021	0.947368421052632\\
	1.10526315789474	0.957024597840955\\
	1.42105263157895	0.955809063189613\\
	1.73684210526316	0.956193077648127\\
	2.04380692235187	1.26315789473684\\
	2.05263157894737	1.27867122039042\\
	2.36842105263158	1.28635169759803\\
	2.68421052631579	1.30879814953493\\
	2.95435974520192	1.57894736842105\\
	3	1.75011027380406\\
	0.6	29\\
	0.205838711559429	0.947368421052632\\
	0.473684210526316	0.743768940821361\\
	0.789473684210526	0.683849470165809\\
	1.05299263509735	0.947368421052632\\
	1.10526315789474	0.994219749750501\\
	1.42105263157895	0.988322031145915\\
	1.73684210526316	0.990185251757444\\
	2.00981474824256	1.26315789473684\\
	2.05263157894737	1.33842784218351\\
	2.36842105263158	1.37569316676266\\
	2.68421052631579	1.4846023262416\\
	2.77855556849525	1.57894736842105\\
	2.68421052631579	1.85832535758637\\
	2.66590591904481	1.89473684210526\\
	2.43154114235413	2.21052631578947\\
	2.36842105263158	2.27571340101341\\
	2.05263157894737	2.48590821542987\\
	1.9192650529926	2.52631578947368\\
	1.73684210526316	2.59200091272653\\
	1.42105263157895	2.61698315061663\\
	1.10526315789474	2.55785198257366\\
	1.03513937938724	2.52631578947368\\
	0.789473684210526	2.42544283991702\\
	0.524898744325865	2.21052631578947\\
	0.473684210526316	2.15997143144023\\
	0.201786545529632	1.89473684210526\\
	0.170556860752044	1.57894736842105\\
	0.160428810739058	1.26315789473684\\
	0.205838711559429	0.947368421052632\\
	1.2	19\\
	0.529167679687639	0.947368421052632\\
	0.789473684210526	0.80834177955438\\
	0.928500325708778	0.947368421052632\\
	1.10526315789474	1.10580520547914\\
	1.42105263157895	1.08586093501482\\
	1.73684210526316	1.09216177408539\\
	1.90783822591461	1.26315789473684\\
	2.05263157894737	1.51769770756279\\
	2.23759567959632	1.57894736842105\\
	2.05263157894737	1.76447945923023\\
	1.75200778488738	1.89473684210526\\
	1.73684210526316	1.90171870756231\\
	1.42105263157895	1.93226750574059\\
	1.28433449891076	1.89473684210526\\
	1.10526315789474	1.85268589178595\\
	0.789473684210526	1.62592551167725\\
	0.751840292552383	1.57894736842105\\
	0.576666759141984	1.26315789473684\\
	0.529167679687639	0.947368421052632\\
	2.2	7\\
	1.30402104020859	1.26315789473684\\
	1.42105263157895	1.24842577479633\\
	1.73684210526316	1.26212264463198\\
	1.73787735536802	1.26315789473684\\
	1.73684210526316	1.2665646910031\\
	1.42105263157895	1.31104768003582\\
	1.30402104020859	1.26315789473684\\
};
\addlegendentry{\scriptsize HK}

\addplot [color=black, draw=none, mark=*, mark options={solid, fill=black, black}, forget plot]
  table[row sep=crcr]{%
0	2\\
};

\addplot [color=black, draw=none, mark size=3.3pt, mark=triangle*, mark options={solid, rotate=180, fill=black, black}, forget plot]
  table[row sep=crcr]{%
-1.5	0\\
1.5	0\\
};

\addplot [color=black, dotted, forget plot]
  table[row sep=crcr]{%
-3	0.866025403784439\\
-1.5	0\\
3	2.59807621135332\\
};

\addplot [color=black, dotted, forget plot]
  table[row sep=crcr]{%
-3	2.59807621135332\\
1.5	0\\
3	0.866025403784439\\
};

\node[fill=white, inner sep=.5pt] at (axis cs: -1.5, 0.3)  {\scriptsize Tx1};
\node[fill=white, inner sep=.5pt] at (axis cs: 1.5, 0.3)  {\scriptsize Tx2};
\node[fill=white, inner sep=.5pt] at (axis cs: -0.3, 2.3)  {\scriptsize Rx1};

\end{axis}

\end{tikzpicture}%
}
			
		}\hspace{1 em}
		\subfloat[Rate region with Rx2 at $(4,4)$ cm]{\label{fig:onChip_G}
			\scalebox{1}{\input{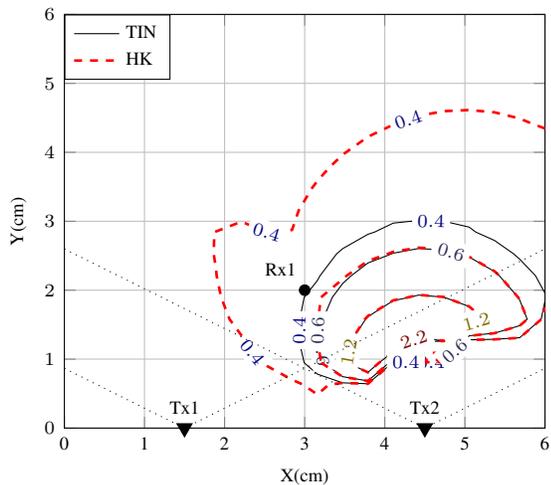}
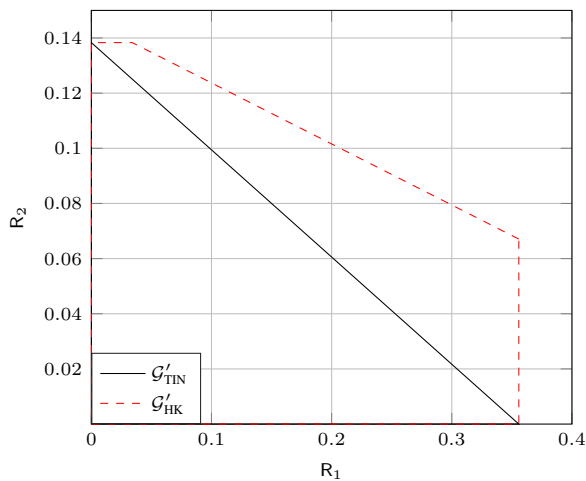}
		}
		\caption{Achievable sum-rate (bit/transmission) for the on-chip optical wireless communication system. Receiver 1 (Rx1) and the two transmitters (Tx1 and Tx2) are represented by the dot and the triangles, respectively, and the dotted lines define the half-intensity range for each transmitter.}
		\label{fig:onChip}
	\end{figure}

	\subsection{Indoor Visible Light Communication (VLC)}
	
	Consider an indoor VLC downlink system in an exhibition room with two light sources that illuminate the room while transmitting data to paired users. The illumination areas of the two light sources overlap to meet illumination requirements, and each light source serves one paired user. Due to the overlap of the illumination areas, interference occurs. To deal with interference, three schemes are considered: TIN, simplified HK, and TDMA with equal time allocation. The achievable sum-rate of TDMA is given by $\frac{1}{2}( \underline{\mathsf{C}}(h_{11}\mathsf{A},\alpha) + \underline{\mathsf{C}}(h_{22}\mathsf{A},\alpha) )$. Note that TDMA requires transmitter-side coordination which is not necessary for TIN and HK schemes.  
	
	Assume that the exhibition room has length, width, and height of $8$, $5$, and $3$ m, respectively, where the origin is the room floor center. Fig. \ref{fig:indoor} depicts the top view of the room. Two identical transmitters modelled as Lambertian emitters are located at $(-2,0,3)$ m and $(2,0,3)$ m, respectively, facing downwards. Two identical receivers are randomly located in the room at $0.8$ m above the ground facing upwards. The channel gain $h_{ij}$, $i,j\in\{1,2\}$, is given by \eqref{eq:lambertian} where we assume \cc{$\mathsf{G}=1$, $\mathsf{S}=10$mm$^2$}, $\Phi = 60^{\circ}$, $\Psi = 70^{\circ}$. We assume that $\mathsf{A}=1000$, $\alpha=0.5$, and $\sigma=1$, and we calculate the sum-rate achieved by TIN and HK using Theorems \ref{thm:TIN-innerB} and \ref{thm:HK-innerB}.
	
	To compare the capabilities of the three schemes under the given settings, we first fix the location of receiver 1 to be at $(-2,1, 0.8)$ m and compare the achievable sum-rate of the schemes with respect to the location of receiver 2 $(x,y,0.8)$ m. The results show that both TIN and HK perform nearly the same in this setting under the given location of receiver 1, thus, we only show the comparison of TDMA and HK in Fig. \ref{fig:indoor}. We can see that the HK scheme can provide a higher sum-rate over a wider area than TDMA, even if receiver 2 is constrained to be in the right half of the room.
	
	Next, we compare the average achievable sum-rate of each scheme with random receiver location. To do so, we divide the room into a $16\times10$ grid of $0.5\times0.5$ m squares, and use the center of each square as a possible location of receivers. Then we divide the room into three regions: left, right, middle, specified by $x<-1$, $x>1$, and $-1\le x \le 1$, respectively, and we assume that receiver 1 is located randomly in the left or middle regions, and receiver 2 is located randomly in the middle or right regions. If we denote the locations of the two receivers using a pair (region 1, region 2), then we can sort these pairs (qualitatively) in increasing order of interference strength as (left, right), then (left, middle) or (middle, right), then (middle, middle). Evaluating the average sum-rate achieved by each scheme in each of these pairs, we obtain the results shown in Table \ref{tab:indoor}. We can observe that when interference is relatively weak, i.e., (left, right), then the TIN and HK schemes achieve the same average sum-rate but outperforms TDMA. As interference gets stronger, i.e., (left, middle) or (middle, right) and (middle, middle), the average sum-rate of the schemes decreases, but the HK scheme outperforms the rest. Thus, the HK scheme shows the best interference management performance. 
	
	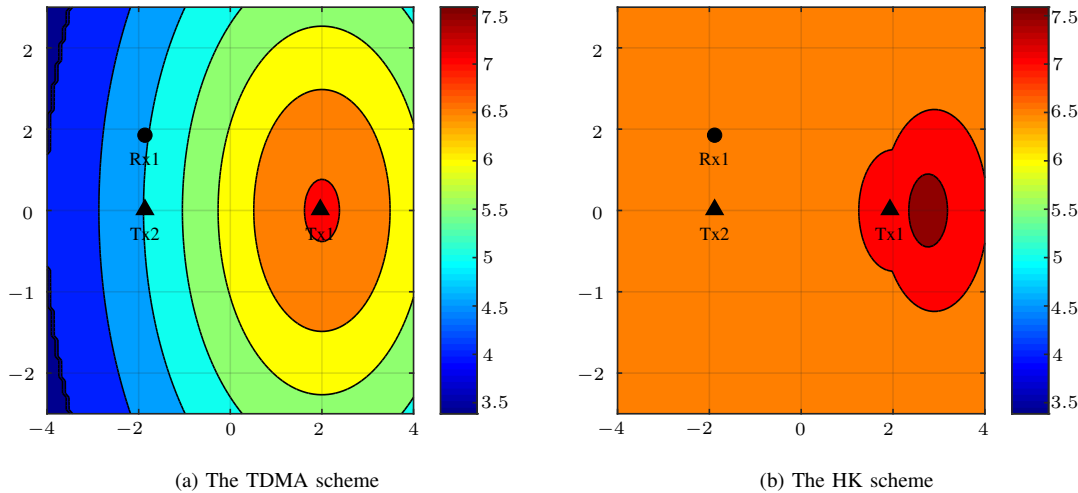
\begin{figure*}
		\centering
		\subfloat[The TDMA scheme]{
			\centering
			\begin{tikzpicture}[scale=.9]
			\pgftext{\includegraphics[width=.9\columnwidth, height= .8\columnwidth]{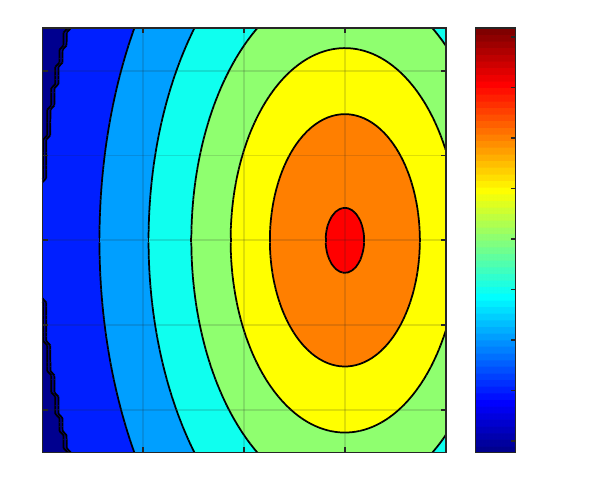}} at (0,0);
			\node[fill=black, regular polygon, regular polygon sides=3,inner sep=1.5pt] at (0.62,0.15) {};
			\node[fill=black, regular polygon, regular polygon sides=3,inner sep=1.5pt] at (-1.97,0.15) {};
			\filldraw[fill=black] (-1.97,1.25) circle (3pt);
			\node[fill=none,inner sep=.5pt] at (0.62,-.2) {\scriptsize Tx1};
			\node[fill=none,inner sep=.5pt] at (-1.97,-.2) {\scriptsize Tx2};
			\node[fill=none,inner sep=.5pt] at (-1.97,.9) {\scriptsize Rx1};
			\node[fill=none,inner sep=.5pt] at (-3.5,-3.1) {\scriptsize $-4$}; 
			\node[fill=none,inner sep=.5pt] at (-2.2,-3.1) {\scriptsize $-2$};
			\node[fill=none,inner sep=.5pt] at (-.7,-3.1) {\scriptsize $0$};
			\node[fill=none,inner sep=.5pt] at (.6,-3.1) {\scriptsize $2$};
			\node[fill=none,inner sep=.5pt] at (2,-3.1) {\scriptsize $4$};
			\node[fill=none,inner sep=.5pt] at (-3.8,-2.3) {\scriptsize $-2$}; 
			\node[fill=none,inner sep=.5pt] at (-3.8,-1.1) {\scriptsize $-1$};
			\node[fill=none,inner sep=.5pt] at (-3.7,.1) {\scriptsize $0$};
			\node[fill=none,inner sep=.5pt] at (-3.7,1.3) {\scriptsize $2$};
			\node[fill=none,inner sep=.5pt] at (-3.7,2.5) {\scriptsize $2$};
			\node[fill=none,inner sep=.5pt] at (3.2,-2.7) {\scriptsize $3.5$}; 
			\node[fill=none,inner sep=.5pt] at (3.1,-2) {\scriptsize $4$};
			\node[fill=none,inner sep=.5pt] at (3.2,-1.3) {\scriptsize $4.5$};
			\node[fill=none,inner sep=.5pt] at (3.1,-.55) {\scriptsize $5$};
			\node[fill=none,inner sep=.5pt] at (3.2,.15) {\scriptsize $5.5$};
			\node[fill=none,inner sep=.5pt] at (3.1,.9) {\scriptsize $6$};
			\node[fill=none,inner sep=.5pt] at (3.2,1.6) {\scriptsize $6.5$};
			\node[fill=none,inner sep=.5pt] at (3.1,2.3) {\scriptsize $7$};
			\node[fill=none,inner sep=.5pt] at (3.2,3) {\scriptsize $7.5$};
			\end{tikzpicture}
		}
		\vspace{2em}
		\subfloat[The HK scheme]{
			\centering
			\begin{tikzpicture}[scale=.9]
			\pgftext{\includegraphics[width=.9\columnwidth, height= .8\columnwidth]{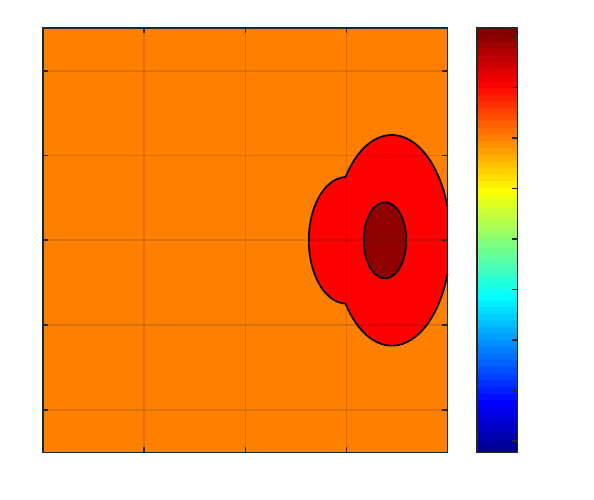}} at (0,0);
			\node[fill=black, regular polygon, regular polygon sides=3,inner sep=1.5pt] at (0.62,0.15) {};
			\node[fill=black, regular polygon, regular polygon sides=3,inner sep=1.5pt] at (-1.97,0.15) {};
			\filldraw[color=black] (-1.97,1.25) circle (3pt);
			\node[fill=none,inner sep=.5pt] at (0.62,-.2) {\scriptsize Tx1};
			\node[fill=none,inner sep=.5pt] at (-1.97,-.2) {\scriptsize Tx2};
			\node[fill=none,inner sep=.5pt] at (-1.97,.9) {\scriptsize Rx1};
			\node[fill=none,inner sep=.5pt] at (-3.5,-3.1) {\scriptsize $-4$}; 
			\node[fill=none,inner sep=.5pt] at (-2.2,-3.1) {\scriptsize $-2$};
			\node[fill=none,inner sep=.5pt] at (-.7,-3.1) {\scriptsize $0$};
			\node[fill=none,inner sep=.5pt] at (.6,-3.1) {\scriptsize $2$};
			\node[fill=none,inner sep=.5pt] at (2,-3.1) {\scriptsize $4$};
			\node[fill=none,inner sep=.5pt] at (-3.8,-2.3) {\scriptsize $-2$}; 
			\node[fill=none,inner sep=.5pt] at (-3.8,-1.1) {\scriptsize $-1$};
			\node[fill=none,inner sep=.5pt] at (-3.7,.1) {\scriptsize $0$};
			\node[fill=none,inner sep=.5pt] at (-3.7,1.3) {\scriptsize $2$};
			\node[fill=none,inner sep=.5pt] at (-3.7,2.5) {\scriptsize $2$};
			\node[fill=none,inner sep=.5pt] at (3.2,-2.7) {\scriptsize $3.5$}; 
			\node[fill=none,inner sep=.5pt] at (3.1,-2) {\scriptsize $4$};
			\node[fill=none,inner sep=.5pt] at (3.2,-1.3) {\scriptsize $4.5$};
			\node[fill=none,inner sep=.5pt] at (3.1,-.55) {\scriptsize $5$};
			\node[fill=none,inner sep=.5pt] at (3.2,.15) {\scriptsize $5.5$};
			\node[fill=none,inner sep=.5pt] at (3.1,.9) {\scriptsize $6$};
			\node[fill=none,inner sep=.5pt] at (3.2,1.6) {\scriptsize $6.5$};
			\node[fill=none,inner sep=.5pt] at (3.1,2.3) {\scriptsize $7$};
			\node[fill=none,inner sep=.5pt] at (3.2,3) {\scriptsize $7.5$};
			\end{tikzpicture}
		}
		\caption{Achievable sum-rate (bit/transmission) vs. the position of receiver 2, where the two transmitters (triangles) are at $(-2,0,3)$ m, $(2,0,3)$ m and receiver 1 (circle) is at $(-2,1,0.8)$ m. The maximum and minimum sum-rate achieved by the TDMA scheme are $7$ and $3.4$ bit/transmission, respectively, compared to $7.6$ and $6.8$ bit/transmission for the HK scheme.}  
		\label{fig:indoor}
	\end{figure*}

	\begin{table}	
		\centering
		\caption{Average achievable sum-rate of the TDMA, TIN, and HK schemes for different users' locations.}
		\label{tab:indoor}
		\begin{tabular}{ c| c c c}
			 & (left, right) & \multirow{2}{2cm}{(left, middle) or (middle, right)} & (middle, middle) \\ &&& \\ \hline
			TDMA & 6 & 5.5 & 5\\  
			TIN & 6.63 & 6.09 & 5.38\\
			HK & 6.63 & 6.15 & 5.59
		\end{tabular}
	\end{table}

	\section{Conclusion}
	\label{Sec:Conclusion}
	We studied the two-user IM/DD IC under peak and average optical intensity constraints, and provided computable capacity inner bounds based on TIN and HK schemes, and outer bounds based on Z- and genie-aided channels. 
	We also approximated the high-SNR asymptotic capacity region in the strong interference regime and the sum-rate capacity in the weak interference regime. Furthermore, we characterized the GDoF for the symmetric IM/DD IC using the obtained inner and outer bounds, which confirms the asymptotic tightness of the obtained capacity bounds.  
	Numerical results shows that TIN can achieve the same good performance as the HK scheme when interference is weak, and that the HK scheme can achieve the capacity region within $0.5$ bits in the strong interference regime at high SNR. We also showed how the results can be used to optimize receivers' location in an on-chip OWC system and to evaluate the interference management performance of different schemes in an indoor VLC system. 
	
	This work can be extended in several directions. One direction involves bounding the capacity for the IM/DD IC with more than two transmitter-receiver pairs or where the transmitters and receivers have multiple optical emitters/detectors (IM/DD MIMO IC). Also, since all the bounds obtained in this paper are based on an IM/DD P2P channel capacity upper bound, tightening the obtained IM/DD IC capacity bounds using tighter P2P capacity upper bounds is a direct extension. Tighter bounds can also be investigated using methods like those used for studying the IM/DD BC or IM/DD MAC \cite{zhou2019bounds}.

	\section*{Acknowledgments}
	We acknowledge the feedback from the editor and the reviewers, whose comments helped to improve the quality of this paper.
	
	\bibliographystyle{IEEEtran}
	\bibliography{references}
	
	\begin{IEEEbiography}[]{Zhenyu~(Charlus)~Zhang} was born in Xinyang, Henan Province, P.R. China. He received the M.Sc. degree in Tianjin University, China, in 2017, and is now a Ph.D. candidate in University of British Columbia, Canada, since 2018. He was awarded Chinese National Scholarship in 2012 and 2016, and he won the Canadian Workshop of Information Theory paper award in 2019. His research interests are in the areas of information theory and wireless communications. 
		
	\end{IEEEbiography}
	
	\begin{IEEEbiography}[]{Anas Chaaban} (S'09 - M'14 - SM'17) received the Ma{\^i}trise {\`e}s Sciences degree in electronics from Lebanese University, Lebanon, in 2006, the M.Sc. degree in communications technology and the Dr. Ing. (Ph.D.) degree in electrical engineering and information technology from the University of Ulm and the Ruhr-University of Bochum, Germany, in 2009 and 2013, respectively. From 2008 to 2009, he was with the Daimler AG Research Group On Machine Vision, Ulm, Germany. He was a Research Assistant with the Emmy-Noether Research Group on Wireless Networks, University of Ulm, Germany, from 2009 to 2011, which relocated to the Ruhr-University of Bochum in 2011. He was a PostDoctoral Researcher with the Ruhr-University of Bochum from 2013 to 2014, and with King Abdullah University of Science and Technology from 2015 to 2017. He joined the School of Engineering at the University of British Columbia as an Assistant Professor in 2018. His research interests are in the areas of information theory and wireless communications.
	\end{IEEEbiography}

\end{document}